\pdfoutput=1
\documentclass[journal]{IEEEtran}

\usepackage[T1]{fontenc}
\usepackage[utf8]{inputenc}
\usepackage{amsmath,amssymb,amsfonts,mathtools}
\usepackage{amsthm}
\usepackage{booktabs}
\usepackage{array}
\usepackage{graphicx}
\usepackage{tikz}
\usetikzlibrary{arrows.meta,positioning,calc,fit,backgrounds,shapes.geometric,
                decorations.pathmorphing}
\usepackage{url}
\usepackage[hidelinks]{hyperref}

\newtheorem{theorem}{Theorem}
\newtheorem{corollary}{Corollary}
\newtheorem{proposition}{Proposition}
\newtheorem{definition}{Definition}
\newtheorem{assumption}{Assumption}
\theoremstyle{remark}
\newtheorem{remark}{Remark}

\newcommand{\Reach}{\mathrm{Reach}}
\newcommand{\Sset}{\mathcal{S}}
\newcommand{\Rset}{\mathcal{R}}
\newcommand{\Tset}{\mathcal{T}}
\newcommand{\Gam}{\Gamma}

\graphicspath{{figs/}}
\providecommand{\SimAThreeConfigs}{\ensuremath{14}}
\providecommand{\SimAThreeFails}{\ensuremath{5}}

\providecommand{\SimAThreeMinHoldActuators}{\ensuremath{3}}
\providecommand{\SimAThreeRankAtFail}{\ensuremath{4}}

\providecommand{\SimAttackQ}{\ensuremath{2}}

\providecommand{\SimBudgetExactFitted}{\ensuremath{3}}
\providecommand{\SimBudgetExactTrue}{\ensuremath{1}}
\providecommand{\SimBudgetFloorFitted}{\ensuremath{1}}
\providecommand{\SimBudgetFloorTrue}{\ensuremath{1}}
\providecommand{\SimCVImprovement}{\ensuremath{10}}

\providecommand{\SimClosedFallbackStep}{\ensuremath{5}}

\providecommand{\SimConfinedFloor}{\ensuremath{1.0\times 10^{-3}}}
\providecommand{\SimConfinedHelps}{\ensuremath{3}}

\providecommand{\SimConfinedQAgnostic}{\ensuremath{2}}
\providecommand{\SimConfinedQConfined}{\ensuremath{6}}
\providecommand{\SimConfinedRefNet}{\ensuremath{2}}
\providecommand{\SimConfinedRefQ}{\ensuremath{1}}
\providecommand{\SimConfinedRows}{\ensuremath{4}}
\providecommand{\SimConfinedTrustedMin}{\ensuremath{3}}
\providecommand{\SimConfinedUntrusted}{\ensuremath{6}}
\providecommand{\SimCorrectedProducts}{\ensuremath{0.0132, 0.0150, 0.0147, 0.0152}}

\providecommand{\SimDeltaMinCV}{\ensuremath{0.516}}
\providecommand{\SimDeltaMinCVCorrected}{\ensuremath{0.054}}

\providecommand{\SimDeltaMinProducts}{\ensuremath{0.0021, 0.0058, 0.0057, 0.0109}}

\providecommand{\SimDeltaMinRes}{\ensuremath{0.5}}
\providecommand{\SimDeltaMinSeedSpread}{\ensuremath{4}}

\providecommand{\SimEnsembleAggPositivePct}{\ensuremath{78}}
\providecommand{\SimEnsembleCellPositivePct}{\ensuremath{56}}
\providecommand{\SimEnsembleChainPct}{\ensuremath{0}}
\providecommand{\SimEnsembleConfigs}{\ensuremath{32}}
\providecommand{\SimEnsembleFamilies}{\ensuremath{3}}
\providecommand{\SimEnsembleGainAggMax}{\ensuremath{6}}
\providecommand{\SimEnsembleGainAggMedian}{\ensuremath{4}}
\providecommand{\SimEnsembleGainCellMedian}{\ensuremath{1}}
\providecommand{\SimEnsembleRandomPct}{\ensuremath{92}}

\providecommand{\SimGammaQHigh}{\ensuremath{0.987}}
\providecommand{\SimGammaQLow}{\ensuremath{0.715}}

\providecommand{\SimGammaSpan}{\ensuremath{9}}

\providecommand{\SimHiddenChannels}{\ensuremath{\{1, 2, 4, 5\}}}

\providecommand{\SimIdErrA}{\ensuremath{1.44\times 10^{-3}}}
\providecommand{\SimIdMarginAgreement}{\ensuremath{1.2}}
\providecommand{\SimIdSamples}{\ensuremath{400}}
\providecommand{\SimIdSpuriousK}{\ensuremath{3}}
\providecommand{\SimIdWorstSpurious}{\ensuremath{8.8\times 10^{-4}}}
\providecommand{\SimInflBad}{\ensuremath{1.0\times 10^{0}}}
\providecommand{\SimInflBudget}{\ensuremath{0.1288}}
\providecommand{\SimInflControls}{\ensuremath{4}}
\providecommand{\SimInflRecovers}{\ensuremath{yes}}

\providecommand{\SimLeakAtBoundary}{\ensuremath{0.061}}
\providecommand{\SimLeakPNinety}{\ensuremath{0.120}}
\providecommand{\SimLeakPTen}{\ensuremath{0.022}}
\providecommand{\SimLeakage}{\ensuremath{<10^{-12}}}

\providecommand{\SimMarginFitted}{\ensuremath{0.0766}}
\providecommand{\SimMarginFloor}{\ensuremath{1.0\times 10^{-3}}}
\providecommand{\SimMarginTrue}{\ensuremath{0.0757}}

\providecommand{\SimModelDraws}{\ensuremath{400}}
\providecommand{\SimModelErrHalfPct}{\ensuremath{0.5}}
\providecommand{\SimModelErrPct}{\ensuremath{1}}

\providecommand{\SimNLTOneHeld}{\ensuremath{2.0}}
\providecommand{\SimNLTOneLost}{\ensuremath{5.0}}

\providecommand{\SimNaiveFallbackStep}{\ensuremath{14}}
\providecommand{\SimNaivePeak}{\ensuremath{0.770}}

\providecommand{\SimNoiseFloor}{\ensuremath{0.02}}

\providecommand{\SimOpenEscapeStep}{\ensuremath{16}}
\providecommand{\SimOpenEstimatePeak}{\ensuremath{0.095}}
\providecommand{\SimOpenPeak}{\ensuremath{52.9}}
\providecommand{\SimOpenTrips}{\ensuremath{0}}
\providecommand{\SimOpenUnderstatement}{\ensuremath{65}}
\providecommand{\SimPlaceGreedySize}{\ensuremath{2}}
\providecommand{\SimPlaceOptSize}{\ensuremath{2}}
\providecommand{\SimPlaceRows}{\ensuremath{5}}
\providecommand{\SimPlaceSubPairs}{\ensuremath{172}}
\providecommand{\SimPlaceSubmodular}{\ensuremath{yes}}
\providecommand{\SimPlaceWorstRatio}{\ensuremath{1.00}}
\providecommand{\SimPlantCells}{\ensuremath{6}}
\providecommand{\SimPlantChannels}{\ensuremath{8}}
\providecommand{\SimProbDetect}{\ensuremath{0.92}}
\providecommand{\SimProbDetectHalf}{\ensuremath{0.76}}
\providecommand{\SimQStar}{\ensuremath{1}}
\providecommand{\SimQStarSafetyAggregate}{\ensuremath{3}}
\providecommand{\SimQStarSafetyOneCell}{\ensuremath{2}}
\providecommand{\SimQStarSafetyThreeAlt}{\ensuremath{1}}
\providecommand{\SimQStarTwoAdjacent}{\ensuremath{2}}
\providecommand{\SimQStarTwoSpread}{\ensuremath{1}}

\providecommand{\SimRho}{\ensuremath{1.08}}

\providecommand{\SimRhoThmOne}{\ensuremath{1.48}}

\providecommand{\SimRobustDiamRatio}{\ensuremath{1.000}}

\providecommand{\SimRobustGamma}{\ensuremath{0.913}}
\providecommand{\SimRobustInfConv}{\ensuremath{6.9}}
\providecommand{\SimRobustLimitAgrees}{\ensuremath{yes}}

\providecommand{\SimRobustTau}{\ensuremath{0.30}}
\providecommand{\SimRobustTauZero}{\ensuremath{0.20}}

\providecommand{\SimSafetyMaps}{\ensuremath{6}}
\providecommand{\SimSafetyWitnesses}{\ensuremath{3}}
\providecommand{\SimScaleAgree}{\ensuremath{yes}}
\providecommand{\SimScaleCompared}{\ensuremath{18}}
\providecommand{\SimScaleEnumSubsets}{\ensuremath{10^{149}}}
\providecommand{\SimScaleLargestCells}{\ensuremath{500}}
\providecommand{\SimScaleLargestChannels}{\ensuremath{503}}
\providecommand{\SimScaleLargestQ}{\ensuremath{125}}
\providecommand{\SimScaleLargestSeconds}{\ensuremath{183}}

\providecommand{\SimSetValuedFallbackStep}{\ensuremath{13}}
\providecommand{\SimSetValuedPeak}{\ensuremath{0.580}}

\providecommand{\SimSupportTrue}{\ensuremath{2}}

\providecommand{\SimTrustAgents}{\ensuremath{5}}
\providecommand{\SimTrustControlOk}{\ensuremath{yes}}
\providecommand{\SimTrustEscapeStep}{\ensuremath{16}}
\providecommand{\SimTrustMinWeight}{\ensuremath{0.48}}
\providecommand{\SimTrustNetChannels}{\ensuremath{0}}
\providecommand{\SimTrustPeak}{\ensuremath{35.1}}
\providecommand{\SimTrustWeightConsistent}{\ensuremath{1.00}}
\providecommand{\SimTrustWeightWrong}{\ensuremath{0.02}}

\providecommand{\SimUnpropEscapeStep}{\ensuremath{16}}

\providecommand{\SimXMax}{\ensuremath{1}}

\newif\ifArtifactReleased
\ArtifactReleasedfalse

\newcommand{\ArtifactRepoURL}{https://github.com/NeuroBazar-Inc/verified-floor}
\newcommand{\ArtifactArchiveDOI}{10.0000/UNASSIGNED}

\newcommand{\ArtifactLocation}{%
  \ifArtifactReleased
    The artifact is at \url{\ArtifactRepoURL} and is archived under
    DOI~\href{https://doi.org/\ArtifactArchiveDOI}{\ArtifactArchiveDOI},
    which is the version this paper's figures were produced from.
  \else
    The artifact is being prepared for public release and will be archived
    under a persistent identifier before publication. It is available from
    the author on request.
  \fi
}

\newcommand{\ArtifactTool}{%
  The assessment is also available as a command. Given a deployment graph,
  a measurement map and an actuation bound, it returns the tolerable
  adversary budget at a stated margin floor, the reaction ceiling in steps
  and in milliseconds, the minimum observation window and the lag it owes,
  the number of concurrent applications the loop may host, and whether the
  deployment has crossed the confinement threshold. Every quantity is
  printed with the assumption that produced it, and any figure that is a
  bound rather than a certificate says so.

  A suite of named reference deployments ships with it, each with its
  computed answers, so a proposal for a better decoder, a tighter bound or
  a different architecture has somewhere to report numbers. One entry has
  a healthy rank-test budget whose robust precondition nonetheless fails,
  which turns the caution about margin floors into something a reader can
  run into rather than read about.%
}

\begin{document}

\title{Runtime Assurance Under Measurement Attack: Necessary and
Sufficient Observability Conditions for Learned Control in Radio
Access Networks}

\author{Yasser~Al-Eryani,~\IEEEmembership{Member,~IEEE,}%
\thanks{Y. Al-Eryani is with NeuroBazar Inc., Canada
(e-mail: yasser.aleryani@neurobazar.ca).}%
}

\markboth{}%
{Al-Eryani: Runtime Assurance Under Measurement Attack}

\maketitle

\begin{abstract}
Runtime assurance pairs a verified fallback with an untrusted controller
and a switching monitor, and is the leading route to admitting learned
policies into safety-relevant network control. Its guarantee rests on a
condition its statements assume rather than require: that the monitor's
estimation error is zero, or at worst stochastic with a characterisable
rate. In a mobile network many of its measurements originate at untrusted
endpoints, where the error is neither. We call the missing condition the
\emph{assurance precondition}, and prove it: at zero measurement noise a
monitor that interrupts the learned controller only once the state has left
its trigger set keeps the plant safe against an adversary controlling $q$
channels if and only if the plant is $2q$-sparse observable with respect to
the safety-relevant output.
That is functional observability quantified over attack supports, and it is
strictly weaker than full-state observability: on our topology a per-cell
outage condition doubles the tolerable budget and one stated over total
offered load triples it. The converse is a constructive attack rendering a safe and
an unsafe trajectory observationally identical, forcing every monitor on
those channels, not a particular detector, to give up safety or
availability, and it does not depend on the noise level. The quantifier over supports is also weaker than it looks:
the trust split is static and known before deployment, so confining the
adversary to one family yields a \emph{confinement threshold}, above which
the trusted channels resolve the state alone and every untrusted channel
may be corrupt at once. Three uninfluenceable counters cross it here,
lifting the budget from two to six and making placement the dominant lever.
Below the precondition the monitor faces a frontier, not a dichotomy; a
set-valued monitor is optimal, and it is also what carries sufficiency into
non-zero noise, at the residual test's false-rejection rate. An adversary
reaching the trigger radius owns the switch regardless. The same argument
bounds any monitor reading agent messages rather than measurements.
Finally, a rank test decides the budget and identification does not
preserve it, so a published budget must name its margin floor.
\end{abstract}

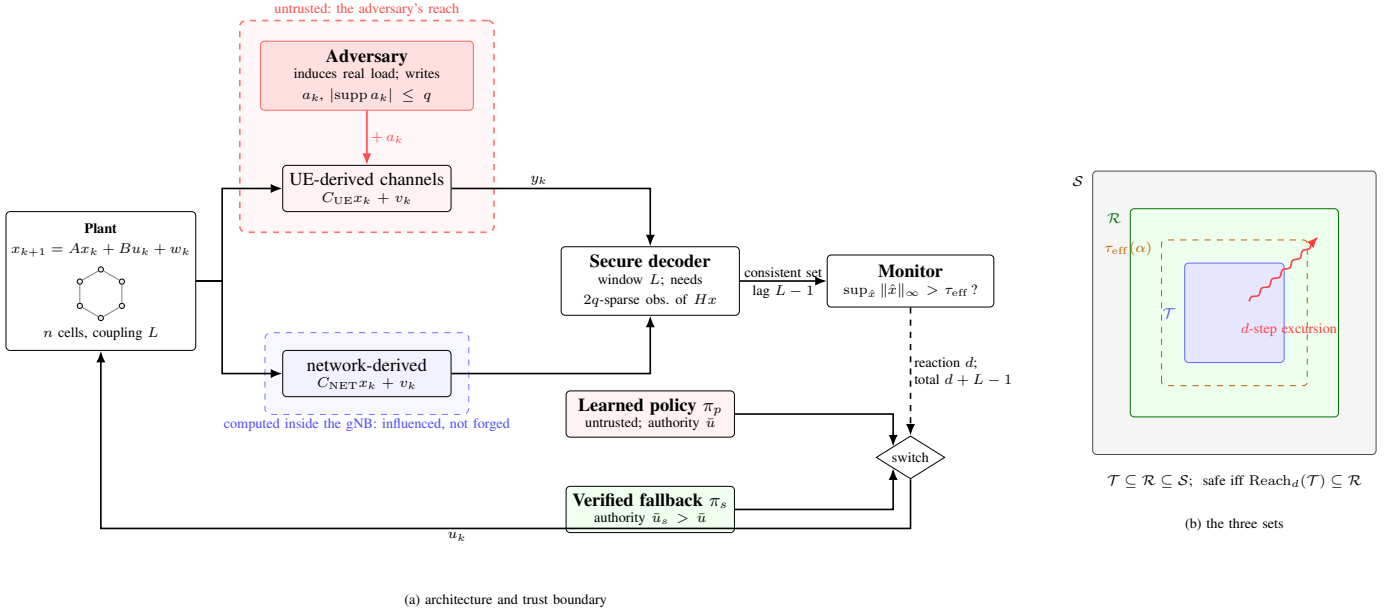
\begin{figure*}[!t]
\centering
\resizebox{\textwidth}{!}{%
\begin{tikzpicture}[
  font=\small,
  blk/.style = {draw, rounded corners=2pt, align=center, inner sep=4pt,
                text width=26mm, minimum height=9mm, fill=white},
  sml/.style = {draw, rounded corners=2pt, align=center, inner sep=3pt,
                text width=27mm, minimum height=8mm, fill=white},
  ar/.style  = {-{Latex[length=2mm,width=1.6mm]}, thick},
  dar/.style = {-{Latex[length=2mm,width=1.6mm]}, thick, dashed},
  lbl/.style = {inner sep=1.5pt, font=\scriptsize},
]

\node[blk, text width=30mm, minimum height=24mm] (plant) at (0,0) {};
\node[lbl, anchor=north] at ([yshift=-1.5mm]plant.north) {\textbf{Plant}};
\node[lbl, anchor=north] at ([yshift=-5.0mm]plant.north)
     {$x_{k+1}=Ax_k+Bu_k+w_k$};
\begin{scope}[shift={($(plant.center)+(0,-0.22)$)}]
  \foreach \i in {0,...,5}
    \node[circle, draw, fill=black!8, inner sep=0.9pt] (c\i)
          at ({90-60*\i}:0.42) {};
  \foreach \i/\j in {0/1,1/2,2/3,3/4,4/5,5/0}
    \draw[gray, line width=0.4pt] (c\i) -- (c\j);
\end{scope}
\node[lbl, anchor=south] at ([yshift=1.2mm]plant.south)
     {$n$ cells, coupling $L$};

\node[sml, fill=red!5]  (ue)  at (4.6, 1.60)
     {UE-derived channels\\[-2pt]{\scriptsize $C_{\mathrm{UE}}x_k+v_k$}};
\node[sml, fill=blue!5] (net) at (4.6,-1.60)
     {network-derived\\[-2pt]{\scriptsize $C_{\mathrm{NET}}x_k+v_k$}};

\draw[ar] (plant.east) -- ++(0.45,0) |- (ue.west);
\draw[ar] (plant.east) -- ++(0.45,0) |- (net.west);

\node[blk, text width=34mm, fill=red!12, draw=red!60] (adv) at (4.6,3.55)
     {\textbf{Adversary}\\[-2pt]{\scriptsize induces real load;
      writes $a_k$, $|\mathrm{supp}\,a_k|\le q$}};
\draw[ar, red!70] (adv.south) -- node[lbl, right, text=red!70] {$+\,a_k$}
      (ue.north);

\begin{scope}[on background layer]
  \node[draw=red!55, dashed, thick, rounded corners=3pt,
        fit=(adv)(ue), inner sep=3.5mm, fill=red!3] (tb) {};
  \node[draw=blue!55, dashed, rounded corners=3pt,
        fit=(net), inner sep=3mm, fill=blue!3] (gb) {};
\end{scope}
\node[lbl, text=red!70, anchor=south] at (tb.north)
     {untrusted: the adversary's reach};
\node[lbl, text=blue!70, anchor=north] at (gb.south)
     {computed inside the gNB: influenced, not forged};

\node[blk, text width=28mm] (dec) at (9.5,0)
     {\textbf{Secure decoder}\\[-2pt]{\scriptsize window $L$; needs
      $2q$-sparse obs.\ of $Hx$}};
\node[blk, text width=26mm] (mon) at (14.0,0)
     {\textbf{Monitor}\\[-2pt]{\scriptsize
      $\sup_{\hat{x}}\|\hat{x}\|_\infty>\tau_{\mathrm{eff}}$\,?}};

\draw[ar] (ue.east)  -| node[lbl, pos=0.22, above] {$y_k$} (dec.north);
\draw[ar] (net.east) -| (dec.south);
\draw[ar] (dec.east) -- node[lbl, above] {consistent set}
                        node[lbl, below] {lag $L-1$} (mon.west);

\node[diamond, draw, aspect=1.6, inner sep=1.5pt, font=\scriptsize]
      (sw) at (14.0,-3.05) {switch};
\draw[dar] (mon.south) -- node[lbl, right, align=left]
      {reaction $d$;\\total $d + L - 1$} (sw.north);

\node[sml, fill=red!5]   (pip) at (9.5,-2.30)
     {\textbf{Learned policy} $\pi_p$\\[-2pt]{\scriptsize untrusted;
      authority $\bar{u}$}};
\node[sml, fill=green!6] (pis) at (9.5,-3.95)
     {\textbf{Verified fallback} $\pi_s$\\[-2pt]{\scriptsize authority
      $\bar{u}_s>\bar{u}$}};

\draw[ar] (pip.east) -| (sw.north west);
\draw[ar] (pis.east) -| (sw.south west);
\draw[ar] (sw.south) |- ++(0,-0.85)
      -| node[lbl, pos=0.28, below] {$u_k$} (plant.south);

\node[lbl, anchor=north] at (7.0,-5.35)
     {(a) architecture and trust boundary};

\begin{scope}[shift={(19.6,-0.55)}]
  \def\Sw{2.45} \def\Rw{1.80} \def\Ew{1.26} \def\Tw{0.86}

  \draw[draw=black!70, rounded corners=2pt, fill=black!4]
       (-\Sw,-\Sw) rectangle (\Sw,\Sw);
  \draw[draw=green!45!black, rounded corners=2pt, fill=green!7]
       (-\Rw,-\Rw) rectangle (\Rw,\Rw);
  \draw[draw=orange!75!black, dashed, rounded corners=2pt]
       (-\Ew,-\Ew) rectangle (\Ew,\Ew);
  \draw[draw=blue!60, rounded corners=2pt, fill=blue!10]
       (-\Tw,-\Tw) rectangle (\Tw,\Tw);

  \node[lbl, anchor=east] at (-\Sw-0.10, \Sw-0.18) {$\mathcal{S}$};
  \node[lbl, anchor=east, text=green!45!black]
       at (-\Rw-0.10, \Rw-0.16) {$\mathcal{R}$};
  \node[lbl, anchor=east, text=orange!75!black]
       at (-\Ew-0.10, \Ew-0.14) {$\tau_{\mathrm{eff}}(\alpha)$};
  \node[lbl, anchor=east, text=blue!70] at (-\Tw-0.10, 0) {$\mathcal{T}$};

  \draw[ar, red!75, decorate,
        decoration={snake, amplitude=0.3mm, segment length=2.4mm,
                    post length=1.5mm}]
       (0.25,0.20) .. controls (0.85,0.70) and (1.15,1.00) .. (1.45,1.32);
  \node[lbl, text=red!75, anchor=north west] at (0.05,-0.12)
       {$d$-step excursion};

  \node[lbl, anchor=north, align=center] at (0,-\Sw-0.28)
       {$\mathcal{T}\subseteq\mathcal{R}\subseteq\mathcal{S}$;\ \
        safe iff $\Reach_d(\mathcal{T})\subseteq\mathcal{R}$};
  \node[lbl, anchor=north] at (0,-\Sw-1.05) {(b) the three sets};
\end{scope}

\end{tikzpicture}}
\caption{The architecture under analysis. (a) A learned policy actuates the
plant, and a monitor decides from the decoder's output whether to hand the
actuator to a verified fallback. The measurement path divides into a family
the adversary can write to and a family computed inside the base station, and
the whole of this paper turns on that division: the decoder resolves the
safety-relevant output only while the plant is $2q$-sparse observable with
respect to $Hx$ (Theorem~\ref{thm:safety}), and the margin of
Definition~\ref{def:margin} measures how much of that observability survives.
The adversary drawn here is tier T2 of \S\ref{sec:model}, which writes
$a_k$ on $q$ channels; the truthful adversary that motivates the problem
moves the state rather than the measurement and is contained by the
architecture. The decoder reads a window of $L$ samples and therefore
resolves the state as it was $L-1$ steps earlier, so its lag is part of the
delay the switch must tolerate (the observation-window corollary of~\cite{companion:delay}). (b) Safety
requires the worst-case excursion over that total delay to remain inside the
recoverable set $\mathcal{R}$. The containment $\mathcal{R}\subset\mathcal{S}$
is strict only when the fallback cannot hold the boundary of $\mathcal{S}$,
which a bounded actuator does not on its own guarantee; the sets are drawn as
boxes for legibility, and~\cite{companion:delay} measures $\mathcal{R}$ as a
body that is not a ball. Tuning the monitor for a false-alarm rate $\alpha$
inflates the trigger from $\mathcal{T}$ to $\tau_{\mathrm{eff}}(\alpha)$ and
consumes the same budget.}
\label{fig:system}
\end{figure*}

\section{Introduction}

\IEEEPARstart{M}{achine} learning is entering mobile network control at
three distinct places, and the security literature has largely treated
them as one problem. 3GPP has studied learned functions for the air
interface, covering channel state information compression, beam management
and positioning \cite{3gpp38843}. A parallel study line places learned
functions in the radio access network for energy saving, load balancing
and mobility optimisation \cite{3gpp37817}. O-RAN gives those functions a
home: the RAN Intelligent Controller hosts third-party applications on a
control loop of 10\,ms to 1\,s, acting on the
radio access network through the E2 interface \cite{polese2023understanding}.
The three settings have incompatible threat models, and recommendations
derived from their union tend to fit none of them.

This paper concerns the second setting: a learned policy that actuates
under a hard latency budget. The dominant proposal for admitting such a
policy safely is runtime assurance in the Simplex
lineage~\cite{seto1998simplex,sha2001simplicity}, extended in recent years
to neural controllers~\cite{phan2020neural}. One pairs the untrusted
policy with a verified fallback and a monitor that hands control to the
fallback when the state approaches the boundary of a safe operating
region. The appeal is precise. The learned policy is never verified at
all; only the fallback and the switching condition are, and both are small
and deterministic. In a control loop that has no budget for runtime
explanation of a decision, this is close to the only tractable posture.

\subsection{The gap}

Every statement of the runtime assurance guarantee we are aware of assumes
that the decision module evaluates its trigger predicate on the true state
\cite{seto1998simplex,sha2001simplicity,phan2020neural,hobbs2023rta}, or on
an estimate whose error is bounded and
stochastic~\cite{kaczmarek2026neuromorphic}. In the deployments this
architecture is proposed for, neither holds. A monitor in the radio access
network reads quantities derived from handset measurement reports, and
handsets are the adversary's cheapest asset. The difference between a
stochastic error and an adversarial one is the subject of this paper: the
first has a rate and admits a risk bound, and the second has neither.

The strongest form of the attack requires no forgery. A device that
truthfully reports a channel condition it has deliberately induced
violates no integrity property, so every cryptographic mechanism in the
stack passes while the control loop is compromised. Authentication of a
channel does not confer truthfulness of its content, and control loops
consume content rather than channels. This paper asks what that costs and
when it can be recovered.

\subsection{Contributions}

\begin{enumerate}
\item \emph{A precondition, stated and proved}
(Theorem~\ref{thm:one}). At zero measurement noise, a monitor that
interrupts the untrusted controller only once the state has left the trigger
set keeps the runtime assurance safety guarantee against an adversary
controlling $q$ measurement channels if and only if the plant is $2q$-sparse observable with respect to
the safety-relevant state. Sufficiency reduces to secure state estimation;
necessity we prove by an indistinguishability construction, and it holds at
every noise level because the two trajectories it builds produce identical
measurements. Remark~\ref{rem:noise} states what sufficiency becomes when
the noise is not zero, and why the equivalence cannot be claimed there.

\item \emph{The weaker condition the architecture actually needs}
(Theorem~\ref{thm:safety}). The monitor does not need the state, only which
side of the safety boundary the plant is on. Stating the safe set through a
safety map $H$ gives a strictly weaker requirement, verified to reproduce
Definition~\ref{def:sparseobs} exactly at $H = I$ and measured to raise the
tolerable budget on our reference plant from $\SimQStar$ to
$\SimQStarSafetyOneCell$ for a per-cell outage condition and to
$\SimQStarSafetyAggregate$ for one stated over total offered load. How much
it buys is a property of the safety map, and \S\ref{sec:safetyobs} reports
the maps where it buys nothing.

\item \emph{A constructive converse} (\S\ref{sec:attack}). The attack
produces two trajectories, one inside the safe set and one leaving it, whose
measurement sequences no function of the available channels can separate. It
therefore forces every monitor reading them, including one with a perfect
model, unbounded computation and full knowledge of the attack strategy,
either to interrupt a trajectory that is safe or to miss one that is not. No
comparison against a baseline detector is required, and none would
strengthen the claim.

\item \emph{Confinement, and the threshold it hides}
(Proposition~\ref{prop:confined}). The corrupted support is not free to be
anywhere: the split between handset-derived and network-derived measurement
is static and known before deployment. Confining it weakens the condition,
and the weakening is a discontinuity rather than a slope. Once the trusted
channels render the state observable on their own, observability stops
bounding the budget at all and every untrusted channel may be corrupt at
once. On our topology the crossing is at $\SimConfinedTrustedMin$ trusted
counters and takes the budget from $\SimConfinedQAgnostic$ to
$\SimConfinedQConfined$. The configuration used throughout this paper sits
below that crossing, and we report it rather than choose one that passes.

\item \emph{A safety--availability frontier, and its optimum}
(Corollaries~\ref{cor:dichotomy} and~\ref{cor:setvalued}). Below the
threshold a monitor that acts on a member of the consistent set can lose
safety and one that declines spends availability, but a monitor that acts on
the whole set does neither: it cannot fail open, and it spends availability
only when the adversary drives the state near the boundary. What survives is
that an adversary reaching the trigger radius owns the switch regardless.

\item \emph{A bound on any message-based trust monitor}
(Proposition~\ref{prop:trust}). A monitor whose statistic reads agent
messages and models of those agents, and nothing else, has no channel
outside the adversary's reach and therefore a tolerable budget of zero. We
implement such a metric and run the construction of \S\ref{sec:attack}
against it.

\item \emph{A correction to the central computable quantity}
(\S\ref{sec:identified}). The tolerable budget is decided by an exact rank
test and exact rank deficiency is not generic, so on a model fitted from
traces the reported budget moves by a factor of three at a fraction of a per
cent of model error while the margin moves by one per cent. Every published
budget must therefore name the margin floor it was computed at.

\item \emph{A detection scaling law} (\S\ref{sec:refuted}). The minimum
detectable state deviation scales inversely with the sparse observability
margin once the margin is taken at a $q$-removal and the shape of the hidden
direction is accounted for: the corrected product has a coefficient of
variation of $\SimDeltaMinCVCorrected$, against $\SimDeltaMinCV$ for the
natural but mis-specified pairing with $\gamma_{2q}$. The pairing error is
one that any similar evaluation can make.
\end{enumerate}

Table~\ref{tab:novelty} places each contribution against the closest prior
work and states the extension in each case.

The reaction-delay side of the architecture --- what the monitor may do once
it has a correct estimate, how long it may take, and how many concurrent
applications the loop can host --- is treated separately in the companion
paper on the delay budget~\cite{companion:delay}, which takes the
precondition established here as given.

\begin{table*}[!t]
\caption{Contributions against the closest prior work.}
\label{tab:novelty}
\centering
\renewcommand{\arraystretch}{1.25}
\begin{tabular}{@{}p{0.235\textwidth} p{0.30\textwidth} p{0.38\textwidth}@{}}
\toprule
\textbf{Contribution} & \textbf{Closest prior work} & \textbf{Strict extension} \\
\midrule
Runtime assurance is conditional on $2q$-sparse observability
  & Simplex, its neural successors and the safety-filtering family~\cite{seto1998simplex,sha2001simplicity,phan2020neural,alshiekh2018shielding,mehmood2022blackbox,hobbs2023rta}, stated under full state observation; and~\cite{kaczmarek2026neuromorphic}, which admits an imperfect monitor whose errors are \emph{stochastic} and bounds risk by their rates
  & An adversary has no rate. The precondition is stated and proved, and the guarantee is shown to fail discontinuously rather than degrade when it is violated \\
The condition is on the safety-relevant \emph{output}
  & Secure state estimation~\cite{fawzi2014secure,chong2015observability,shoukry2016event,shoukry2017smt,lee2019redundant}, which asks whether the \emph{state} is recoverable; functional observability~\cite{montanari2022functional,montanari2025pbh}, which asks about a functional but not about an adversary
  & The two questions asked together: functional observability quantified over attack supports. Strictly weaker than full-state observability, reduces to it exactly at $H = I$, and raises the tolerable budget threefold on our topology \\
A confined support, and the threshold it crosses
  & Secure state estimation quantifies over every support of size $2q$, since the corrupted set is taken to be arbitrary~\cite{fawzi2014secure,shoukry2017smt}; resilient placement is formulated for full-state recovery~\cite{shinohara2024placement}
  & The trust split is static and known, so the quantifier runs over one family; once the trusted rows observe the state alone the budget is unbounded within the other, which turns placement into the dominant lever and changes the placement objective from recovery to a threshold \\
Constructive indistinguishability attack on the \emph{monitor}
  & Undetectable false data injection in another domain~\cite{liu2011falsedata}; robustness bounds for attack-resilient estimation~\cite{pajic2014robustness}
  & Lifts a recovery condition to a safety architecture; the attack forces every monitor on those channels to give up safety or availability, not a particular detector \\
A frontier with a named optimum, at zero measurement noise
  & Detection-theoretic trade-offs in intrusion detection; recovery after detection~\cite{zhang2021recovery}
  & Structural rather than statistical, and the set-valued monitor is shown to dominate both alternatives while still conceding the switch \\
A bound on any message-based trust monitor
  & Consistency-scored multi-agent governance~\cite{chergui2026agents}; watermarking~\cite{mo2009replay} and barrier functions under sensor attack~\cite{clark2020cbf}
  & Such a monitor reads no channel outside the adversary's reach, so its tolerable budget is zero regardless of its accuracy on honest faults \\
The tolerable budget is not identification-robust
  & O-RAN security corpus~\cite{oran_wg11_secwork,oran_zta_2024}, organised around authorisation and model protection
  & An exact rank test moves threefold under identification noise while the margin does not, so a budget is meaningless without its floor \\
\bottomrule
\end{tabular}
\end{table*}
\section{Related work}

\subsection{Runtime assurance}

The Simplex architecture originates with Seto \emph{et
al.}~\cite{seto1998simplex} and Sha~\cite{sha2001simplicity}, and has been
developed extensively for avionics and autonomous vehicles; Hobbs \emph{et
al.}~\cite{hobbs2023rta} survey the modern family under the heading of
safety filtering. Phan \emph{et al.}~\cite{phan2020neural} extend it to
neural controllers, which is the form relevant here, and Mehmood \emph{et
al.}~\cite{mehmood2022blackbox} remove the requirement that the baseline
controller be statically verified by deciding the switch at runtime. The
learning literature reaches the same shape from the other side: shielded
reinforcement learning~\cite{alshiekh2018shielding} composes an untrusted
policy with a correctness envelope computed offline.

The monitor's own inputs have been treated, and it is worth being exact
about how. Kaczmarek~\cite{kaczmarek2026neuromorphic} models the monitor as
a stochastic sensor with a measured false-negative rate, false-positive
rate and detection-latency distribution, and derives a probabilistic bound
on failure risk over a finite horizon. That is the strongest existing
statement of the guarantee under an imperfect monitor, and it is a
statement about a monitor whose errors are stochastic and characterisable.
The distinction this paper turns on is that an adversary is neither. A
stochastic error model admits a bound because errors are drawn; the
construction of \S\ref{sec:attack} produces two trajectories whose
measurement sequences are identical, so no rate characterises it and no
horizon bound follows. Our claim is therefore not that the guarantee has
only been stated under perfect observation --- it has not --- but that
every statement of it we are aware of assumes the monitor's estimation
error is either zero or stochastic, and that adversarial corruption of the
monitor's inputs changes the failure from a rate into a discontinuity.

Two adjacent lines do treat corrupted measurement. Control barrier functions
have been extended to sensor faults and attacks~\cite{clark2020cbf}, and
attack recovery~\cite{zhang2021recovery} takes the opposite posture to
containment, restoring a trajectory after a detected attack rather than
bounding one during it. Neither asks what we ask: not whether a safe control
exists under corrupted measurement, but whether the switching condition of a
runtime-assurance architecture can be evaluated at all. Adversarial
corruption of the monitor's own inputs, in that architecture, is to our
knowledge not treated.

\subsection{Secure state estimation}

Fawzi \emph{et al.}~\cite{fawzi2014secure} establish that state
reconstruction under $q$ corrupted sensors is possible if and only if the
system is $2q$-sparse observable. Chong \emph{et
al.}~\cite{chong2015observability} characterise observability under
adversarial attack, and Shoukry and Tabuada~\cite{shoukry2016event} give
event-triggered observers for the same setting. Lee \emph{et
al.}~\cite{lee2019redundant} unify several of these notions under
redundant observability, connecting the security index, attack
detectability and observability under attack; the condition below is a
member of that family and we state it in their terms where they differ. Pajic \emph{et
al.}~\cite{pajic2014robustness} extend the recovery guarantee to noise and
modelling error, which is the bound the inflated trigger of~\cite{companion:delay} rests on, and
Shoukry \emph{et al.}~\cite{shoukry2017smt} give a satisfiability-based
decoder that is the tractable alternative to the exhaustive enumeration we
use. The undetectable false data injection of Liu \emph{et
al.}~\cite{liu2011falsedata} is the closest structural analogue to our
construction in another domain, and the watermarking line beginning with Mo
and Sinopoli~\cite{mo2009replay} is the standard answer to it, one that
requires an actuation channel the monitor controls rather than a measurement
one it merely reads. That literature asks
whether the state can be recovered. We connect it to whether a safety
architecture holds, which is a different question with a different failure
mode: a recovery result that degrades gracefully corresponds here to a
guarantee that fails discontinuously.

\subsection{Functional observability, and where our condition sits}
\label{sec:rel:funcobs}

The condition this paper proves is not about recovering the state. It is
about recovering enough of the state to decide which side of a boundary the
plant is on, which is a question about a linear function of the state, and
that question has its own literature. Montanari \emph{et
al.}~\cite{montanari2022functional} define functional observability --- the
recoverability of a target functional rather than of the full state --- and
give a graph-based theory with scalable algorithms for the minimal sensor
set; the eigenvalue characterisation is completed
in~\cite{montanari2025pbh}, which establishes the generalised
Popov--Belevitch--Hautus test for the property.

Definition~\ref{def:hobs} is that property quantified over attack supports.
We say so plainly, because the contribution here is not the condition but
the composition: that a runtime-assurance monitor needs functional and not
full observability, that the requirement is therefore strictly weaker than
the secure-estimation literature assumes when it is imported into this
architecture, and that the difference is worth a measured factor in
tolerable adversary budget. Neither literature makes that connection.
Secure state estimation asks for the state and quantifies over supports;
functional observability asks for a functional and does not consider an
adversary. The safety-relevant question needs both at once.

Two consequences follow that we flag here and return to. First, the
enumeration we use to decide Definition~\ref{def:hobs} is exponential in
the channel count, and the tests of~\cite{montanari2025pbh} are not; a
decision procedure built on them would remove the scale limit that confines
our sweeps to small deployments, and we have not built one. Second, the
placement question of Corollary~\ref{cor:design} is a resilient sensor
placement problem, for which Shinohara and
Namerikawa~\cite{shinohara2024placement} give an optimal formulation for
secure state estimation. Our placement result is an observation about one
topology and not an algorithm, and theirs is the formulation an operator
should reach for once the objective is the threshold of
Proposition~\ref{prop:confined}(ii) rather than full-state recovery.

\subsection{O-RAN security}

The O-RAN ALLIANCE security working group maintains the normative security
corpus for this architecture across thirteen work items, spanning the SMO,
the Near-RT RIC, O-Cloud, the open fronthaul, certificate management, an
OAuth~2.0 framework, application lifecycle management, log management and
security testing, and organised around a zero trust
architecture~\cite{oran_wg11_secwork,oran_zta_2024}. Independent surveys
cover the same ground~\cite{polese2023understanding,liyanage2023openran}.

AI/ML is among those work items, and has been since 2024. Its stated goal
is ``to secure AI/ML across the O-RAN architecture, including the SMO,
Non-RT RIC, Near-RT RIC, rApps, and xApps'', pursued by developing
``controls to protect against attacks on AI/ML used in O-RAN architectural
elements''~\cite{oran_wg11_secwork}. The resulting controls concern model
integrity and confidentiality, protection of training data against
poisoning, adversarial training, and differential privacy.

Every one of these is a protection, access-control or lifecycle concern.
Each asks whether a component or an asset is authentic, authorised,
confidential or intact. None asks whether a control action is correct for
the state the network is actually in, and the gap is structural rather
than a matter of coverage: authorisation is a static property of the
actor while correctness is a dynamic property of the situation. A
compromised but authorised application emits perfectly authorised actions
indefinitely, and a model whose integrity is fully intact can still be
wrong. Securing a model against attack and containing a model that
misbehaves are different problems, and the corpus addresses the first.

The framing is visible in the working group's zero trust study, in which
the terms \emph{conflict}, \emph{fallback}, \emph{control loop} and
\emph{measurement} do not occur. Learned components appear there as an
instrument of security and not as a subject of it: RICs and their
applications ``can leverage visibility with artificial intelligence (AI)
to provide continuous monitoring and dynamic policy''~\cite{oran_zta_2024}.
Deconfliction between concurrent applications is specified, but in the Near-RT RIC architecture, not in the security corpus~\cite{polese2023understanding}, and it arbitrates contention between applications and does not validate any application's action against the state of the plant.

\subsection{Attacks on the near-real-time control loop}

Closest to our threat model, Alimohammadi \emph{et
al.}~\cite{alimohammadi2024kpi} poison the key performance indicators that
feed the Near-RT control loop and propose an LSTM detector, reporting
detection rates between 62\% and 99\% depending on how far the injected
values are amplified. Their adversary spoofs traffic on the E2 interface
or compromises an E2 node, so the injected values are forgeries, and a
sufficiently sensitive detector can in principle find them.

Our adversary is weaker in one respect and stronger in another. It need
not forge anything, since load it genuinely induces is reported
truthfully, and it writes on a bounded set of channels, not an interface. Corollary~\ref{cor:dichotomy} then bounds what any detector of
that kind can achieve: below the threshold the benign and poisoned
hypotheses are observationally identical, so no detector separates them at
any accuracy or any amount of training data. The two results are complementary, not opposed. Theirs measures how well a detector
performs inside a regime where detection is possible; ours delimits that
regime.

\subsection{Adversarial machine learning in wireless systems}

A substantial body of work demonstrates adversarial perturbations against
learned physical-layer functions, including attacks that survive a
realistic channel~\cite{kim2022channel}. That work attacks the policy. We
attack the supervisor, which is the component whose purpose is to contain
an untrustworthy policy.

\section{System and threat model}
\label{sec:model}

\subsection{Plant}

We take a linearised discrete-time model of a cell-load control loop of
the kind studied for energy saving and load
balancing~\cite{3gpp37817}:
\begin{align}
  x_{k+1} &= A x_k + B u_k + w_k, & \|w_k\|_\infty &\le \bar{w},
    \label{eq:plant}\\
  y_k &= C x_k + v_k + a_k, & v_k &\sim \mathcal{N}(0,\sigma^2 I).
    \label{eq:meas}
\end{align}
Here $x_i$ is cell $i$'s load deviation from its target operating
point, normalised so that $|x_i| \ge 1$ denotes overload. We take
$A = (1-\lambda)I + cL$, with $L$ the deployment graph Laplacian, since
load shed by one cell arrives at its neighbours, and $\lambda$ a natural
decay term. The spectral radius $\rho$ of $A$ is set by $\lambda$, and
the delay ceiling of~\cite{companion:delay} depends on its value.

Two plants are used, and we distinguish them because the distinction is
not cosmetic. Results about Theorem~\ref{thm:one} --- the construction of
\S\ref{sec:attack}, the frontier of Table~\ref{tab:frontier}, and the
model-knowledge remark --- use the more expansive plant, $\rho =
\SimRhoThmOne$, where an unmonitored deviation grows fast enough to leave
$\Sset$ within a short horizon. Results about
the containment condition of~\cite{companion:delay} --- the ceiling, the feasibility region and the
necessity demonstration --- use $\rho = \SimRho$, where the delay
arithmetic of the delay ceiling of~\cite{companion:delay} is not dominated by a single
step. Every figure and table names the configuration it belongs to.

\subsection{Two families of measurement channel}
\label{sec:families}

The rows of $C$ divide into two families, shown in
Fig.~\ref{fig:system}(a). \emph{UE-derived} channels carry load proxies
aggregated from handset measurement reports. An adversary with devices in
a cell shapes these, and can do so without falsifying any field.
\emph{Network-derived} channels carry quantities such as resource block
utilisation and admission counters, computed inside the base station from
state the endpoint does not author. An endpoint may influence these by
offering real traffic; it cannot forge them.

\subsection{Adversary}

Three adversaries appear in discussions of this architecture and they are
routinely conflated. We separate
them, because our results cover exactly one.

\begin{IEEEdescription}[\IEEEsetlabelwidth{T3.}]
\item[T1.] \emph{The truthful adversary} owns devices and induces real
conditions.
It forges nothing: a handset that truthfully reports a channel condition it
deliberately caused violates no integrity property, so every cryptographic
mechanism in the stack passes. \emph{This adversary moves the state, not the
measurement}, so the monitor reads the truth and the architecture contains
it. Its role here is to explain why owning measurement channels is cheap,
and no result below is about it.

\item[T2.] \emph{The channel adversary} corrupts a fixed set $\Gam$ with
$|\Gam| \le q$ using arbitrary unbounded values, and influences the untrusted
controller's output. \emph{This is the adversary every theorem below
concerns.} The last capability models a compromised application, a poisoned
model, or a policy whose training objective was shaped upstream.

\item[T3.] \emph{The infrastructure adversary} alters $A$, $B$ or $C$,
which
corresponds to a compromised base station rather than compromised handsets.
Strictly stronger, and \emph{none of our results apply to it}.
\end{IEEEdescription}

The separation matters because the motivation and the analysis have different
subjects. T1 is what makes the problem worth posing; T2 is what is proved. A
construction that writes $a_k$ is a T2 construction however the introduction
motivates it, and~\cite{companion:delay} shows that even within T2 the standard
convention of a fixed $\Gam$ is doing more work than it appears to.

We also do not require the adversary to keep an observable performance
indicator healthy while attacking; omitting that constraint makes our attack
faster than a realistic one, which is conservative in the defender's favour.

\subsection{Safe, recoverable and trigger sets}

Write $\pi_p$ for the untrusted learned controller, with per-channel
actuation bound $\bar{u}$, and $\pi_s$ for the verified fallback, with
bound $\bar{u}_s > \bar{u}$. The inequality is a requirement rather than a
convention: a fallback with no more authority than the policy it replaces
cannot undo what that policy did.

$\Sset = \{x : \|x\|_\infty \le x_{\max}\}$ denotes the safe set.
$\Rset = \{x : \|x\|_\infty \le r\}$ denotes the \emph{recoverable} set,
from which $\pi_s$, respecting $\bar{u}_s$, returns the state to the origin
without leaving $\Sset$. $\Tset = \{x : \|x\|_\infty \le \tau\}$ denotes
the trigger set, within which the monitor stays quiet. These satisfy
$\Tset \subseteq \Rset \subseteq \Sset$, equivalently
$\tau \le r \le x_{\max}$, as in Fig.~\ref{fig:system}(b). The containment
$\Rset \subseteq \Sset$ is strict exactly when the fallback cannot hold
the boundary of $\Sset$, which for the plant of \eqref{eq:plant} means
$\bar{u}_s < (\rho-1)\,x_{\max}$. A bounded actuator is therefore not on
its own enough to separate the two sets, and~\cite{companion:delay} reports
a configuration of ours in which they coincide.

\begin{remark}
\label{rem:RvsS}
Conflating $\Rset$ with $\Sset$ is a common error in applied runtime
assurance. A state inside $\Sset$ from which the fallback cannot recover
is already lost, and a monitor that trips there is late by construction.
\cite{companion:delay} exhibits such a state and shows the fallback
failing from it.
\end{remark}

\section{The precondition}
\label{sec:thm1}

\begin{definition}[$k$-sparse observability]
\label{def:sparseobs}
$(A,C)$ is $k$-sparse observable if for every index set $\Gam$ with
$|\Gam| = k$, the pair $(A, C_{\Gam^c})$ is observable, where $C_{\Gam^c}$
denotes $C$ with the rows indexed by $\Gam$ removed.
\end{definition}

\begin{definition}[Sparse observability margin]
\label{def:margin}
Let $O(A,C) = [\,C^\top\ (CA)^\top\ \cdots\ (CA^{n-1})^\top\,]^\top$ be the
observability matrix. The margin at sparsity $k$ is
\begin{equation}
  \gamma_k = \min_{|\Gam|=k}
    \sigma_{\min}\!\left(O(A, C_{\Gam^c})\right),
  \label{eq:margin}
\end{equation}
and $\gamma_k > 0$ if and only if $(A,C)$ is $k$-sparse observable.
\end{definition}

We write $q^\star$ for the largest adversary budget the plant tolerates,
that is, the largest $q$ such that $(A,C)$ is $2q$-sparse observable.

Theorem~\ref{thm:one} requires only the binary property.
The composition analysis of~\cite{companion:delay} requires the margin, because under
measurement noise the amount of observability that survives determines the
detection power available to the monitor.

\begin{assumption}[Reachability of the hidden direction]
\label{asm:a3}
Let $z$ be a hidden direction, that is $z \in \mathcal{N}(A, C_{\Gam^c})$
for some $|\Gam| \le 2q$. There exist a scalar $\beta$ and an input
sequence admissible for the untrusted controller,
$\|u_j\|_\infty \le \bar{u}$, driving the state from the origin to
$\beta z \notin \Sset$ while the disturbance respects $\bar{w}$. Without
this an unobservable direction exists but cannot be occupied, and the
converse below has no construction behind it.
\end{assumption}

The assumption is not free, and where it binds is measured rather than
speculated about. With $B$ a scaled identity every direction is reachable at
sufficient authority and the assumption holds trivially, which is the case
in every other configuration studied here. The realistic case is a
deployment with fewer actuators than cells.

Table~\ref{tab:a3} removes them one at a time on the six-cell ring. With one
or two actuators the reachable subspace has rank $\SimAThreeRankAtFail$ and
the hidden direction the converse needs lies \emph{outside} it: the
direction exists, it is invisible to the monitor, and the untrusted
controller cannot occupy it. Assumption~\ref{asm:a3} fails there, so the
necessity half of Theorem~\ref{thm:one} has no construction behind it and
the deployment is safe for a reason this analysis does not model. From
$\SimAThreeMinHoldActuators$ actuators the reachable subspace is the whole
state space and the assumption holds.

The shape of that gap points the reassuring way. An under-actuated
deployment is not exposed by a hidden direction it cannot reach; what it
loses is the \emph{converse}, not the guarantee. Across
$\SimAThreeConfigs$ configurations the assumption fails in
$\SimAThreeFails$, so the sweep exercises it in both directions instead of
confirming it everywhere, which is all a trivially satisfied hypothesis can
ever do.

\begin{table}[!t]
\caption{Assumption~\ref{asm:a3} exercised. Under-actuation is the realistic
case, and one a scaled-identity input matrix never exercises; where the
hidden direction is unreachable the converse does not apply.}
\label{tab:a3}
\centering
\small
\begin{tabular}{cccc}
\toprule
actuators & reachable rank & hidden direction & Assumption~\ref{asm:a3} \\
\midrule
1 & 4 & \textbf{unreachable} & \textbf{fails} \\
2 & 4 & \textbf{unreachable} & \textbf{fails} \\
3 & 6 & reachable & holds \\
4 & 6 & reachable & holds \\
5 & 6 & reachable & holds \\
6 & 6 & reachable & holds \\
\bottomrule
\end{tabular}

\end{table}

\begin{theorem}[Runtime assurance under corrupted measurement]
\label{thm:one}
Consider the architecture of Fig.~\ref{fig:system}(a) with zero reaction
delay, under Assumption~\ref{asm:a3}, and let the measurement of
\eqref{eq:meas} be noiseless, $\sigma = 0$. A monitor that hands control to the fallback only
once the state has left $\Tset$ guarantees $x_k \in \Sset$ for all $k$, for
every behaviour of the untrusted controller and against every attack with
$|\Gam| \le q$, if and only if $(A,C)$ is $2q$-sparse observable.
\end{theorem}

The restriction on the monitor is what the equivalence is about. Without it,
safety alone can always be bought with availability: a monitor that trips
whenever any consistent state is unsafe never fails open, whatever $(A,C)$
is (Corollary~\ref{cor:setvalued}, Theorem~\ref{thm:robust}). What the
precondition decides is whether safety can be had \emph{without} interrupting
a plant that is still inside its trigger set.

\begin{remark}[Why the hypothesis is $\sigma = 0$ and not a technicality]
\label{rem:noise}
The plant of \eqref{eq:meas} carries $v_k \sim \mathcal{N}(0,\sigma^2 I)$,
and the theorem cannot be stated over that model without the qualifier. At $\sigma > 0$ the noise is unbounded, so no
finite observation separates the two candidates with certainty and no
guarantee of the form ``$x_k \in \Sset$ for all $k$'' holds almost surely
at any level of sparse observability. Sparse observability is therefore
necessary at every noise level and sufficient only at $\sigma = 0$.

What survives at $\sigma > 0$ is weaker and is stated where it is used
rather than smuggled in here. Sufficiency becomes a set-valued statement:
the decoder returns a consistent set, Corollary~\ref{cor:setvalued} shows
the supremum over that set cannot fail open whenever the truth lies inside
it, and the truth lies inside it except on the residual test's
false-rejection event. Safety then inherits that event's probability,
which is the rate the test is sized for, and the trigger inflates by the
bound of~\cite{companion:delay} whose noisy estimation guarantee is due to
Pajic \emph{et al.}~\cite{pajic2014robustness}. Necessity is unaffected:
the construction of Appendix~\ref{app:thm1} produces two trajectories with
\emph{identical} measurement sequences, and adding noise to both changes
nothing about their indistinguishability.

We separate the two rather than assert an equivalence over a noisy model,
because the direction that fails under noise is the one an architect would
rely on.
\end{remark}

\begin{IEEEproof}
($\Leftarrow$) By \cite{fawzi2014secure}, $2q$-sparse observability implies
that the state is uniquely decodable from measurements of which at most
$q$ channels are corrupted. The monitor therefore evaluates its trigger
predicate on the true state, and the classical argument applies:
$\Tset \subseteq \Rset$ and forward invariance of $\Rset$ under $\pi_s$
give $x_k \in \Sset$ for all $k$.

($\Rightarrow$) If $(A,C)$ is not $2q$-sparse observable there is a
direction whose entire output history is confined to $2q$ channels. An
adversary owning half of them cancels its own half and lets the remainder
be attributed to a hypothetical attack on the half it does not own,
leaving the monitor with two states, one safe and one not, that are
consistent with the same observations. The benign one remains in $\Tset$, so
a monitor that trips only once the state has left $\Tset$ cannot trip on it,
and therefore cannot trip on the truth either. The construction is given in
full in Appendix~\ref{app:thm1}.
\end{IEEEproof}

\subsection{The condition the architecture actually needs}
\label{sec:safetyobs}

Theorem~\ref{thm:one} asks for $2q$-sparse observability of the full state.
The monitor does not need the full state. It needs to know which side of the
safety boundary the plant is on, and that is a question about an output.
Write the safe set as $\Sset = \{x : \|Hx\|_\infty \le x_{\max}\}$ for a
safety map $H$; the manuscript's set is the case $H = I$, and in a
deployment it is usually not that case, because an operator's outage
condition is stated over the cells that carry the service rather than over
every cell in the model.

\begin{definition}[$H$-sparse observability]
\label{def:hobs}
$(A,C)$ is $H$-sparse observable at level $k$ if for every $\Gam$ with
$|\Gam| = k$ the unobservable subspace $\mathcal{N}(A, C_{\Gam^c})$ is
contained in $\ker H$.
\end{definition}

Since $\ker I = \{0\}$, Definition~\ref{def:hobs} at $H = I$ is
Definition~\ref{def:sparseobs}. The margin extends with it. Writing $O_\Gam = O(A, C_{\Gam^c})$, the
safety-relevant margin at level $k$ is
\begin{equation}
  \gamma^H_k = \min_{|\Gam| = k}\ \min_{\|Hv\| = 1}\ \min_{w \in \ker H}
      \bigl\| O_\Gam (v + w) \bigr\|,
  \label{eq:hmargin}
\end{equation}
where the innermost minimum removes whatever a component in $\ker H$ can
absorb, since such a component costs nothing in safety terms and may cancel
observability that the safety map never needed. At $H = I$ the kernel is
trivial and \eqref{eq:hmargin} is \eqref{eq:margin}. As with $\gamma_k$,
$\gamma^H_k > 0$ if and only if $(A,C)$ is $H$-sparse observable at
level~$k$.

\begin{theorem}[The assurance precondition]
\label{thm:safety}
Under Assumption~\ref{asm:a3} and the noise hypothesis of
Theorem~\ref{thm:one}, a monitor that hands control to the fallback only
once the safety output has left the trigger set guarantees $x_k \in \Sset$
against every attack with $|\Gam| \le q$ if and only if $(A,C)$ is
$H$-sparse observable at level $2q$. We call this condition the \emph{assurance precondition}
and refer to it by that name throughout.
\end{theorem}

Remark~\ref{rem:noise} applies here unchanged, and for the same reason:
necessity is a statement about two identical measurement sequences and
does not depend on the noise level, while sufficiency at $\sigma > 0$ is
the set-valued statement of Corollary~\ref{cor:setvalued} rather than this
equivalence.

\begin{IEEEproof}
Sufficiency: any two states consistent with the observations differ by an
element of $\mathcal{N}(A, C_{\Gam'^c})$ for some $|\Gam'| \le 2q$, and by
hypothesis $H$ annihilates that difference, so every consistent state agrees
on membership of $\Sset$ and the monitor's predicate is well defined on the
consistent set. The classical argument then applies unchanged. Necessity:
if some hidden direction $z$ has $Hz \neq 0$, scale it as in
Appendix~\ref{app:thm1}; the two candidates then disagree on membership of
$\Sset$ while producing identical measurements.
\end{IEEEproof}

Two things make this a generalisation rather than a restatement, and both
are measured. At $H = I$ the two conditions agree at every level in both the
boolean and the margin, and the agreement is verified numerically at
every level rather than argued. And the gap is real: over $\SimSafetyMaps$
safety maps on the plant of \S\ref{sec:attack},
Table~\ref{tab:safety} shows $\SimSafetyWitnesses$ for which the tolerable
budget rises above $\SimQStar$ while ordinary sparse observability fails.
How far it rises is a property of the map and not a single number. A
per-cell outage condition, over one cell or over two on opposite sides,
reaches $\SimQStarSafetyOneCell$. A condition stated over the total offered
load, which is how an energy-saving objective is usually written, reaches
$\SimQStarSafetyAggregate$: the aggregate is blind to the whole of a hidden
direction that a per-cell condition still partly sees, and that is the
mechanism rather than a coincidence of this ring.

\begin{remark}[A tolerance that cannot fire at rank one]
\label{rem:rank1tol}
The rank-one rows of this table sit on a numerical trap that is worth a
paragraph, because its shape recurs.

The exhaustive test decides $H$-sparse observability from the margin
\eqref{eq:hmargin}, and snaps a margin below the arithmetic's resolution to
zero, because a boolean derived from a $4\times 10^{-16}$ margin reports the
floating-point format rather than the plant. A natural guard takes that
tolerance relative to the largest singular value of the matrix under test.
When $H$ has rank one that matrix has a single column, so its largest and
smallest singular values are the same number and the test reduces to
$x \le x \cdot \epsilon$, which is false for every positive $x$. \emph{Such a
guard cannot fire at $\mathrm{rank}\,H = 1$}, which is exactly the case a
per-cell safety condition produces, and it reports observability at margins
of $4\times 10^{-16}$ and $2\times 10^{-16}$; on this table it would report
$q^\star = 3$ for both per-cell maps.

The trap is exposed by a second, independent decision procedure --- the
spectral route of \S\ref{sec:spectral} --- and by requiring the two to agree:
with a relative tolerance they disagree on precisely the rank-one rows. We
take the tolerance from the scale of the observability matrix, which makes it
independent of $\mathrm{rank}\,H$, and the values printed above are computed
that way.

A guard that cannot fire is invisible to every test that exercises it,
including one written specifically to catch a margin being read as
observability. Only a second computation of the same quantity exposes it.
\end{remark}

The relaxation is not free, and the rows where it buys nothing are what make
the result a property of the safety map rather than of our construction.
With the safety condition stated over three alternating cells the hidden
subspace carries safety signal again and the budget returns to
$\SimQStarSafetyThreeAlt$. An architect therefore has a third lever
alongside trusted-channel placement and actuation authority: state the
safety condition over what actually matters, and check
Definition~\ref{def:hobs} against it.

\begin{table}[!t]
\caption{The tolerable budget under different safety maps, on the plant of
\S\ref{sec:attack}. The first row is the control: at $H = I$ the condition
must reproduce Definition~\ref{def:sparseobs} exactly, and it does. Rows in
bold are configurations that are $H$-sparse observable at a level where
ordinary sparse observability fails.}
\label{tab:safety}
\centering
\begin{tabular}{lccc}
\toprule
safety map $H$ & rank & level & $q^\star$ \\
\midrule
identity (the published safe set) & 6 & 3 & 1 \\
one safety-relevant cell & 1 & 5 & \textbf{2} \\
two, opposite sides of the ring & 2 & 5 & \textbf{2} \\
three contiguous & 3 & 3 & 1 \\
three alternating & 3 & 3 & 1 \\
total offered load & 1 & 7 & \textbf{3} \\
\bottomrule
\end{tabular}

\end{table}

\begin{remark}[Why the threshold is $2q$ and not $q$]
The adversary never reproduces the benign measurement sequence, and cannot:
the difference between the two candidates occupies $2q$ channels while the
adversary owns $q$. It cancels only its own half of that difference and
allows the remainder to be attributed to a hypothetical attack on channels
it does not own. This asymmetry is the whole of the mechanism, and it is
why the requirement doubles.
\end{remark}

\subsection{The converse as a construction}
\label{sec:attack}

We realise the converse as an explicit attack, not an existence argument, since the practical question a reader asks is whether an actual
adversary could mount it. In a
$\SimPlantCells$-cell ring deployment with two network-derived channels,
giving $p = \SimPlantChannels$ and a tolerable budget of
$q^\star = \SimQStar$, the construction at $q = \SimAttackQ$ returns a
hidden direction observable only on channels $\SimHiddenChannels$, with
leakage onto channels outside that support of $\SimLeakage$ and residual
support of $\SimSupportTrue$ under both hypotheses.

Leakage is asserted at machine precision before any safety violation is
reported. A construction that produces a violation while leaking onto a
channel the adversary does not own has not defeated the monitor; it has
only not been checked.

\begin{corollary}[Placement of trusted measurement]
\label{cor:design}
If every UE-derived channel is potentially adversarial, the
network-derived channels must render the safety-relevant subspace
observable on their own.
\end{corollary}

This converts the informal advice to add an uninfluenceable measurement
path into a computable requirement. Our sweep shows in addition that
placement matters as much as count, and not in the intuitive direction:
two adjacent network-derived
counters yielded $q^\star = \SimQStarTwoAdjacent$ where two maximally
separated ones yielded $q^\star = \SimQStarTwoSpread$. Sensor placement for sparse observability is a combinatorial problem, not a coverage problem, and intuition about spreading sensors out is not a reliable
guide. We report this as evidence against relying on intuition, not as a
placement heuristic, since a different topology can invert it.

\begin{figure*}[!t]
\centering
\includegraphics[width=\textwidth]{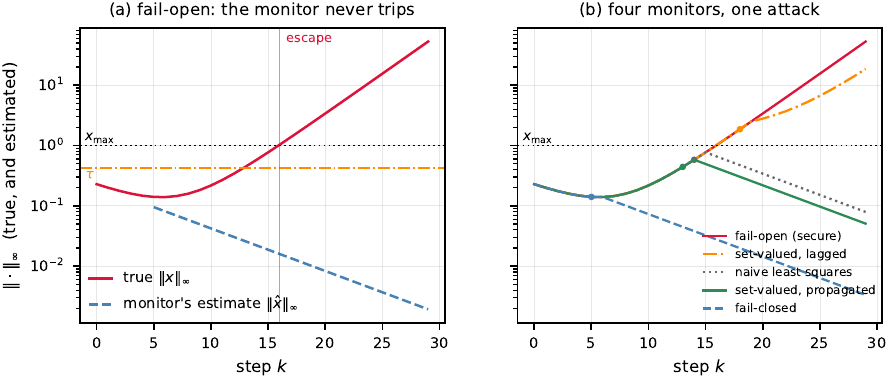}
\caption{Corollary~\ref{cor:dichotomy}, at zero measurement noise.
(a) A monitor that emits an adversary-chosen member of the consistent set.
Its own estimate decays monotonically to $\SimOpenEstimatePeak$ and below,
while the true state leaves the safe set at step $\SimOpenEscapeStep$ and
reaches $\SimOpenPeak$. At the moment of escape the estimate under-reads
the true norm by a factor of $\SimOpenUnderstatement$, and the monitor
trips $\SimOpenTrips$ times in $30$ steps. The plotted estimate is the one
the monitor computes, not a quantity reconstructed by the plotting layer.
(b) The same attack against four monitors. The set-valued monitor of
Corollary~\ref{cor:setvalued} holds the safe set at $\SimSetValuedPeak$ and
is forced into fallback only at step $\SimSetValuedFallbackStep$, against
step $\SimClosedFallbackStep$ for the fail-closed monitor; the same monitor
without propagation across the observation window loses the safe set at
step $\SimUnpropEscapeStep$, which is the observation-window corollary of~\cite{companion:delay} charging
its delay. Trajectories are those of the simulation itself; the plotting
layer performs no computation.}
\label{fig:dichotomy}
\end{figure*}

\begin{corollary}[Safety--availability frontier]
\label{cor:dichotomy}
Below $2q$-sparse observability the decoder holds two consistent states,
one safe and one not. Any monitor must therefore either act on a member of
that set, in which case an adversary that selects the member loses it
safety, or decline to act on a member, in which case it spends
availability. No monitor reading only those channels escapes the exchange.
\end{corollary}

The exchange is not a choice between emitting an estimate and refusing. A
monitor may act on the whole consistent set rather than on a member of it, and doing so is
strictly better than either alternative.

\begin{corollary}[The set-valued monitor is the frontier's optimum]
\label{cor:setvalued}
Let the monitor hand control to the fallback exactly when
$\sup\{\|\hat{x}\|_\infty : \hat{x} \text{ consistent}\} > \tau$.
\emph{Whenever the true state lies in the consistent set}, the monitor does
not lose safety, at any margin, because the supremum bounds the true norm
from above. It spends availability only when some consistent state is
genuinely near the boundary.
\end{corollary}

The hypothesis is not vacuous and it is not automatic. At zero measurement
noise the true state is always consistent, so the corollary applies
unconditionally and that is the regime Table~\ref{tab:frontier} reports. At
non-zero noise it holds exactly when the decoder's residual test admits the
truth, which is a false-rejection event with the probability the test is
sized for; the same union bound that sets $\tau_{\mathrm{eff}}$ in the
companion paper~\cite{companion:delay} sizes it. Safety therefore inherits
the residual test's false-rejection rate, and a monitor built on this
corollary must report that rate rather than claim safety outright. We state
the hypothesis because the alternative --- asserting safety at any noise
level --- would be claiming a guarantee the estimator cannot supply.

Fig.~\ref{fig:dichotomy} shows the two halves of this. Panel (a) is the
monitor Corollary~\ref{cor:dichotomy} concerns, watched against its own
estimate; panel (b) places all four monitors on one axis.

What survives of the impossibility is sharper, and it is what
Table~\ref{tab:frontier} reports. An adversary willing to drive the state
along the hidden direction as far as the trigger radius forces the fallback
whenever it chooses, and no monitor reading these channels prevents that.
The architecture converts a safety failure into an availability cost; it
does not remove the adversary's control of the switch. The absence of noise
is the point: the frontier is not produced by a detection threshold but by
two candidate states being observationally identical, so no reduction in
noise and no increase in computation removes it.

\begin{table}[!t]
\caption{Five monitors against one attack, at zero measurement noise
($x_{\max} = \SimXMax$). The naive least-squares row is current practice
and is \emph{not} the monitor Corollary~\ref{cor:dichotomy} concerns: it is
attenuated rather than deceived, and it does trip. The last row is the
set-valued monitor without propagation across its own observation window,
included as the control for that claim.}
\label{tab:frontier}
\centering
\begin{tabular}{lccc}
\toprule
monitor & left $\mathcal{S}$ & peak $\|x\|_\infty$ & fallback \\
\midrule
fail-open, adversarial choice & yes, step 16 & 52.9 & never \\
naive least squares & no & 0.770 & step 14 \\
fail-closed & no & 0.228 & step 5 \\
\textbf{set-valued} & \textbf{no} & \textbf{0.580} & \textbf{step 13} \\
set-valued, not propagated & yes, step 16 & 18.5 & step 18 \\
\bottomrule
\end{tabular}

\end{table}

Two rows of Table~\ref{tab:frontier} deserve comment. A naive least-squares monitor, given the
same attack, does \emph{not} lose the safe set: it under-reads the true
state but still trips, at step $\SimNaiveFallbackStep$, with a peak of
$\SimNaivePeak$. Deploying it is inadvisable for other reasons, but the
failure Corollary~\ref{cor:dichotomy} describes belongs to the monitor that
holds the consistent set and emits a member of it, and only that monitor
exhibits it here. Second, the set-valued monitor's safety depends on the
estimate being carried forward to the present step; thresholding the raw
decoded state, which is the state as it was at the start of the window,
loses the safe set on this plant.

\begin{remark}[Model knowledge required]
The construction assumes exact knowledge of $(A,C)$. Degrading the
adversary's model makes the hidden direction leak onto channels it does
not own, and how much is a distribution rather than a number: the
perturbation is random, and its draw-to-draw spread is comparable to the
effect. Over $\SimModelDraws$ draws at $\SimModelErrPct$ per cent error
the leakage at the safety boundary has median $\SimLeakAtBoundary$ with
$[\SimLeakPTen, \SimLeakPNinety]$ covering the central eight deciles, against
an assumed noise floor of $\SimNoiseFloor$. Detection therefore occurs in
a fraction $\SimProbDetect$ of draws, not always. At
$\SimModelErrHalfPct$ per cent the fraction is $\SimProbDetectHalf$, so
that is where the defence stops being reliable. Model uncertainty is a
genuine defence against the attack of Section~\ref{sec:attack}, and a
probabilistic one.
\end{remark}

\subsection{Confinement to a known family, and the threshold it crosses}
\label{sec:confined}

Theorem~\ref{thm:one} quantifies over every subset of channels of size
$2q$, because the corrupted support is unknown and could be anywhere. In
this architecture it cannot be anywhere. The division of
\S\ref{sec:families} is static, known before deployment, and not a property
the adversary chooses: a party holding handsets cannot write on a counter
computed inside the base station. Writing $U$ for the untrusted family, the
question an operator actually faces is what budget survives when the
support is known to lie inside $U$.

Two things change, and the second is a discontinuity rather than a
relaxation.

\begin{proposition}[Confinement]
\label{prop:confined}
Let the corrupted support satisfy $\Gam \subseteq U$ with $|\Gam| \le q$.
\begin{enumerate}
\item The guarantee holds if and only if $(A, C_{\Gam'^c})$ is observable
for every $\Gam' \subseteq U$ with $|\Gam'| \le 2q$. The quantifier runs
over subsets of $U$ alone, so the condition is implied by, and strictly
weaker than, $2q$-sparse observability.
\item If $(A, C_{U^c})$ is observable, the condition holds for every $q$ up
to $|U|$. Observability then places no bound on the budget at all.
\end{enumerate}
\end{proposition}

\begin{IEEEproof}
(i) Two states are indistinguishable exactly when their difference lies in
$\mathcal{N}(A, C_{\Gam'^c})$ for $\Gam'$ the union of the two candidate
supports, and that union is contained in $U$ and has size at most $2q$.
Quantifying over subsets of $U$ is therefore exactly the requirement, and
it is a sub-collection of the subsets Definition~\ref{def:sparseobs}
quantifies over. (ii) Every such $\Gam'$ satisfies
$\Gam'^c \supseteq U^c$, so $\mathcal{N}(A, C_{\Gam'^c}) \subseteq
\mathcal{N}(A, C_{U^c}) = \{0\}$ whatever $q$ is.
\end{IEEEproof}

Part (ii) is Corollary~\ref{cor:design} made quantitative, and it is the
reason the placement question is worth computing rather than reasoning
about. We call the condition it states the \emph{confinement threshold}:
the point at which the trusted family renders the state observable without
help. Below it the operator buys budget one channel at a time against a
combinatorial condition; at it the untrusted family stops mattering,
because the decoder can discard all of it and still resolve the state.
There is no intermediate regime in which the budget grows smoothly.

Stated as a placement objective, crossing the confinement threshold is a
resilient sensor placement problem of the kind Shinohara and
Namerikawa~\cite{shinohara2024placement} formulate, with one difference
that matters: their objective is recovery of the state under a budget,
and this one is a threshold beyond which the budget stops binding. We do
not solve the placement problem here, and note it as the form the question
should take.

\begin{table}[!t]
\caption{Confinement against the support-agnostic condition, on the ring of
\S\ref{sec:attack} with $\SimConfinedUntrusted$ handset-derived channels.
``Alone'' is whether the trusted family renders the state observable by
itself, which is Proposition~\ref{prop:confined}(ii). Budgets in bold are
those confinement raises. Every margin is quoted at the floor
$\SimConfinedFloor$; at a zero floor one configuration reads as observable
at a margin of $10^{-5}$, which is observability no measurement noise
delivers.}
\label{tab:confined}
\centering
\small
\begin{tabular}{lccccc}
\toprule
trusted & $p$ & alone & agnostic & confined & margin \\
\midrule
2 & 8 & no & 1 & 1 & 0.632 \\
3 & 9 & yes & 2 & \textbf{6} & 0.600 \\
4 & 10 & yes & 2 & \textbf{6} & 0.600 \\
5 & 11 & yes & 3 & \textbf{6} & 0.771 \\
\bottomrule
\end{tabular}

\end{table}

Table~\ref{tab:confined} measures both parts. With
$\SimConfinedRefNet$ trusted counters the trusted family is not observable
on its own, confinement buys nothing, and the budget stays at
$\SimConfinedRefQ$. With $\SimConfinedTrustedMin$ it is, and the budget
rises from $\SimConfinedQAgnostic$ to $\SimConfinedQConfined$: every
handset-derived channel in the deployment may be corrupt at once, and the
monitor still evaluates its predicate on the true state. Confinement helps
in $\SimConfinedHelps$ of $\SimConfinedRows$ configurations, which is what
makes the sweep a measurement rather than a demonstration.

Two consequences deserve to be separated, because they point in opposite
directions.

The first is a design rule with a number attached. \textbf{Three
uninfluenceable counters are worth more on this topology than any amount of
handset-derived measurement}, and the reason is structural rather than
statistical: they move the deployment across the threshold of
Proposition~\ref{prop:confined}(ii), after which the adversary's budget is
no longer bounded by observability. An operator specifying an energy-saving
or load-balancing loop can check this before deployment, from the
deployment graph and the measurement map alone, and it is the cheapest
intervention available.

The second is a caution about our own configuration, and it is why
$q^\star = \SimQStar$ appears throughout this paper. \textbf{The reference
plant does not satisfy Corollary~\ref{cor:design}}: with
$\SimConfinedRefNet$ network-derived channels the trusted family leaves the
state unobservable, so the architecture is relying on handset-derived
measurement to resolve the safety-relevant subspace, which is precisely the
dependence the corollary tells an architect to remove. We report the plant
as it is rather than choosing one that passes, because a reference
configuration that satisfies its own design rule would conceal how easily
the rule is missed.
\section{The precondition under noise}
\label{sec:robust}

Theorem~\ref{thm:safety} is an equivalence at zero measurement noise, and
Remark~\ref{rem:noise} says why it cannot be one above it. This section
replaces it with a quantitative statement that holds at every noise level,
and the replacement absorbs three other things the paper carried separately:
the margin, the inflated trigger, and the rule that a published budget must
name its floor.

Take the noise to be energy bounded over the decoder's window,
$\|v\|_2 \le \bar{v}$, rather than Gaussian. That is the honest model here.
Gaussian noise is unbounded, so no statement of the form
``$x_k \in \Sset$ for all $k$'' holds almost surely at any level of sparse
observability, which is exactly why Theorem~\ref{thm:one} had to be stated
at $\sigma = 0$.

\begin{theorem}[The robust precondition]
\label{thm:robust}
Let the monitor be the set-valued monitor of
Corollary~\ref{cor:setvalued}, let the untrusted controller be permitted to
run whenever the true safety output satisfies $\|Hx\|\le\tau_0$, and let the
trigger radius be $\tau > \tau_0$. Then under $\|v\|_2 \le \bar{v}$ and a
T2 adversary with $|\Gam| \le q$:
\begin{enumerate}
\item \emph{Safety holds unconditionally.} The true state is consistent with
the observations, so the supremum over the consistent set bounds
$\|Hx\|$ from above and the monitor cannot fail open, whatever
$(A,C)$ is.
\item \emph{Availability holds if and only if}
\begin{equation}
  \gamma^H_{2q} \;\ge\; \frac{2\bar{v}}{\tau - \tau_0}.
  \label{eq:robust}
\end{equation}
\end{enumerate}
\end{theorem}

\begin{IEEEproof}
Two states consistent with the same measurements differ by a $\delta$ whose
residual, on every channel outside the union of the two attack supports and
hence on all but at most $2q$ channels, is at most $2\bar{v}$ in energy. By
the definition of the margin, $\|O_{\Gam^c}\delta\| \ge \gamma^H_{2q}
\|H\delta\|$, so
\begin{equation}
  \|H\delta\| \;\le\; \frac{2\bar{v}}{\gamma^H_{2q}} \;=:\; D,
  \label{eq:diam}
\end{equation}
and $D$ is the diameter of the consistent set measured in the safety output.
For (i), the truth is one of the consistent states, so the supremum is at
least $\|Hx\|$ and the monitor is quiet only while $\|Hx\| \le \tau$. For
(ii), if $\|Hx\| \le \tau_0$ then every consistent $\hat{x}$ satisfies
$\|H\hat{x}\| \le \tau_0 + D$, which is at most $\tau$ exactly when
$D \le \tau - \tau_0$, which is \eqref{eq:robust}. Conversely if
\eqref{eq:robust} fails, a difference direction attaining the margin,
scaled to the admissible residual, produces a consistent state above $\tau$
while the truth is inside $\tau_0$, and the monitor hands over control it
did not need to.
\end{IEEEproof}

Four things follow.

\emph{Theorem~\ref{thm:safety} is the vanishing-noise limit.} As
$\bar{v} \to 0$ the right side of \eqref{eq:robust} goes to zero and the
condition becomes $\gamma^H_{2q} > 0$, which is
Theorem~\ref{thm:safety} exactly. We verify the reduction numerically at
every budget rather than assert it (\SimRobustLimitAgrees).

\emph{Safety was never what the precondition bought.} Part (i) needs no
condition at all. What fails without the precondition is availability: the
consistent set grows without bound, the supremum exceeds any trigger, and
the architecture sits in fallback permanently. That is a sharper statement
of the frontier of Corollary~\ref{cor:dichotomy} than the corollary makes.

\emph{The inflated trigger is \eqref{eq:diam}.} Writing
$\tau_{\mathrm{eff}} = \tau_0 + D$ recovers the form the companion
paper~\cite{companion:delay} derives from a Gaussian quantile; the quantile
and the energy bound are two noise models over one object.

\emph{The margin floor is the right-hand side.} ``Every published budget
must name its floor'' stops being a reporting convention and becomes a
computable quantity, $2\bar{v}/(\tau-\tau_0)$, fixed before deployment. A
budget from an exact rank test is a budget quoted at floor zero, and
\eqref{eq:robust} says that is never the operative floor at any nonzero
noise. This is the theorem \S\ref{sec:identified} was reaching for.

\begin{figure}[!t]
\centering
\includegraphics[width=0.86\columnwidth]{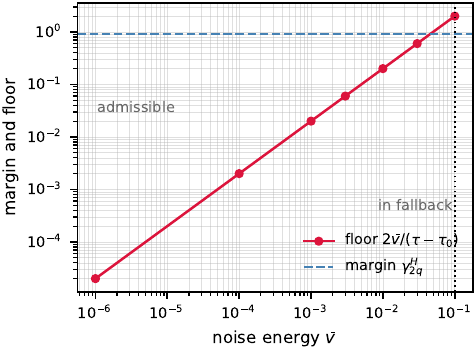}
\caption{The robust precondition \eqref{eq:robust} as a crossing. The floor
$2\bar{v}/(\tau-\tau_0)$ rises with the noise while the margin does not, and
where they cross the architecture stops being able to run the learned
controller at all. The condition is not a property of the plant alone: it is
a comparison between the plant and the noise the deployment actually has.}
\label{fig:robust}
\end{figure}

\begin{table}[!t]
\caption{The crossing of Fig.~\ref{fig:robust} in numbers, on the reference
plant
($\gamma^H_{2q} = \SimRobustGamma$, $\tau = \SimRobustTau$,
$\tau_0 = \SimRobustTauZero$). The floor rises with the noise and the
configuration crosses, which is what makes the sweep a measurement.}
\label{tab:robust}
\centering
\small
\begin{tabular}{ccccc}
\toprule
$\bar{v}$ & floor & $\gamma^H_{2q}$ & diameter & admissible? \\
\midrule
$1.0\times 10^{-6}$ & 0.0000 & 0.913 & 0.0000 & yes \\
$1.0\times 10^{-4}$ & 0.0020 & 0.913 & 0.0002 & yes \\
$1.0\times 10^{-3}$ & 0.0200 & 0.913 & 0.0022 & yes \\
$3.0\times 10^{-3}$ & 0.0600 & 0.913 & 0.0066 & yes \\
$1.0\times 10^{-2}$ & 0.2000 & 0.913 & 0.0219 & yes \\
$3.0\times 10^{-2}$ & 0.6000 & 0.913 & 0.0657 & yes \\
$1.0\times 10^{-1}$ & 2.0000 & 0.913 & 0.2192 & \textbf{no} \\
\bottomrule
\end{tabular}

\end{table}

Table~\ref{tab:robust} sweeps the noise. The diameter bound
\eqref{eq:diam} is searched against sampled worst-case difference
directions rather than trusted: the largest ratio of achieved to bounded
extent is $\SimRobustDiamRatio$, so the bound holds and is close to tight.
Converting to the infinity norm the trigger actually uses costs at most a
factor of $\SimRobustInfConv$, in the direction that enlarges the trigger,
which is the conservative one.

\subsection{Deciding the condition without enumerating supports}
\label{sec:spectral}

Definition~\ref{def:hobs} quantifies over every channel subset of size $k$,
and every result above was computed by enumerating them. That is
$\binom{p}{k}$ observability tests, and it is the single choice with two
costs: it makes the decoder unimplementable in a near-real-time budget, and it
confines every sweep here and in~\cite{companion:delay} to a six-cell ring.

It is avoidable. Write
\begin{equation}
  \mathrm{supp}_{\mathrm{out}}(z)
    = \{\, m : (C A^j z)_m \neq 0 \text{ for some } j \ge 0 \,\}
  \label{eq:outsupp}
\end{equation}
for the channels on which a direction is ever visible. Deleting a set $\Gam$
destroys observability exactly when $\Gam$ contains
$\mathrm{supp}_{\mathrm{out}}(z)$ for some $z \neq 0$, so

\begin{proposition}[The condition is one minimisation]
\label{prop:minsupp}
$(A,C)$ is $k$-sparse observable if and only if
$\min_{z \neq 0} |\mathrm{supp}_{\mathrm{out}}(z)| > k$, and is $H$-sparse
observable at level $k$ if and only if the same minimum, taken over $z$ with
$Hz \neq 0$, exceeds $k$. The minimiser's support is the channel set an
attacker should buy first.
\end{proposition}

There is no enumeration in that statement, and the spectrum evaluates it.
Let $V$ span an eigenspace of $A$ with eigenvalue $\lambda$ and dimension
$d$. For $z = Vc$ a channel $m$ is invisible exactly when a $d$-row block
$B_m c$ vanishes, because the sequences $j \mapsto \binom{j}{i}
\lambda^{j-i}$ are linearly independent. Across eigenspaces nothing is
gained: for $z$ spanning several eigenvalues a Vandermonde argument forces
each component's output to vanish separately, so
$\mathrm{supp}_{\mathrm{out}}$ is the union of the components' supports and
is never smaller than the smallest. Hence
\begin{equation}
  \min_{z \neq 0} |\mathrm{supp}_{\mathrm{out}}(z)|
    = \min_{\text{eigenspaces}} \ \min_{c \neq 0}
      \bigl|\{\, m : B_m c \neq 0 \,\}\bigr|,
  \label{eq:spectral}
\end{equation}
and the inner minimum is attained at a $c$ annihilated by at most $d-1$
channel blocks, so it costs $O(p^{d-1})$ rather than $\binom{p}{k}$.

\emph{The saving is governed by the largest eigenvalue multiplicity, not by
$p$.} For a generic deployment graph $d = 1$ and the cost is linear in the
channel count. On a symmetric ring $d = 2$, and the degeneracy that makes it
two is the same degeneracy that creates the hidden directions. On the
directed chain of~\cite{companion:delay}, whose state matrix has a constant
diagonal and therefore one Jordan block, $d = n$ and the method degenerates
to the enumeration it replaces. We report that rather than claim a universal
speedup, and the implementation refuses the spectral route and calls the
enumeration when the eigenvector basis is too ill-conditioned to trust.

Every answer here is checked against the enumeration on every configuration
where the enumeration still runs, in the boolean and in the margin, and they
agree in all $\SimScaleCompared$ of them (\SimScaleAgree). Beyond that range
the reach is what matters: a $\SimScaleLargestCells$-cell deployment with
$\SimScaleLargestChannels$ channels is decided in
$\SimScaleLargestSeconds$\,s, giving $q^\star = \SimScaleLargestQ$, where the
enumeration would need about $\SimScaleEnumSubsets$ subsets.

\begin{figure}[!t]
\centering
\includegraphics[width=\columnwidth]{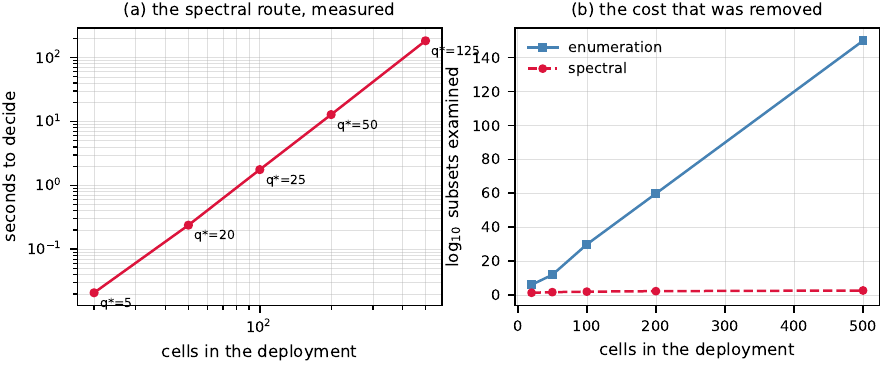}
\caption{Deployments the enumeration cannot reach. (a) Time to decide the
condition exactly, against deployment size, with the tolerable budget
annotated. (b) The cost that was removed: subsets the enumeration must
examine against the eigenspaces the spectral route examines, on a decimal
logarithm. The last row of Table~\ref{tab:scale} would need about
$\SimScaleEnumSubsets$ observability tests by enumeration and takes
$\SimScaleLargestSeconds$\,s here.}
\label{fig:scale}
\end{figure}

\begin{table}[!t]
\caption{The same measurements as Fig.~\ref{fig:scale}(a), for a reader who
wants the numbers.}
\label{tab:scale}
\centering
\small
\begin{tabular}{lccr}
\toprule
cells & channels & $q^\star$ & spectral (s) \\
\midrule
20 & 23 & 5 & 0.02 \\
50 & 53 & 20 & 0.24 \\
100 & 103 & 25 & 1.76 \\
200 & 203 & 50 & 12.83 \\
500 & 503 & 125 & 183.50 \\
\bottomrule
\end{tabular}

\end{table}

This connects the condition to an existing literature rather than only
accelerating it. Equation~\eqref{eq:spectral} is a sparse-support statement
about functional observability, for which eigenvalue tests and graph
algorithms already exist~\cite{montanari2022functional,montanari2025pbh};
what is new here is the quantification over attack supports, and it is that
quantification the reformulation removes.

\subsection{The headline numbers, as distributions}
\label{sec:ensemble}

With the enumeration gone, the single-topology objection can be answered
instead of conceded. We repeat the safety-map measurement over
$\SimEnsembleConfigs$ configurations spanning $\SimEnsembleFamilies$
families: symmetric rings, directed chains, and random geometric graphs
built from cell positions, which is what a deployment graph actually is.

The gain generalises, and not uniformly (Fig.~\ref{fig:ensemble} and
Table~\ref{tab:ensemble}). Stating the safety condition over
total offered load raises the tolerable budget in
$\SimEnsembleAggPositivePct$ per cent of configurations, with a median gain
of $\SimEnsembleGainAggMedian$ and a maximum of
$\SimEnsembleGainAggMax$; a per-cell condition raises it in
$\SimEnsembleCellPositivePct$ per cent with a median gain of
$\SimEnsembleGainCellMedian$. On random geometric graphs some gain appears
in $\SimEnsembleRandomPct$ per cent of cases. On the directed chain it
appears in $\SimEnsembleChainPct$ per cent: the chain is the family where
stating the safety condition more narrowly buys nothing at all, and its
presence is what makes this a measurement rather than a demonstration.

\begin{figure}[!t]
\centering
\includegraphics[width=\columnwidth]{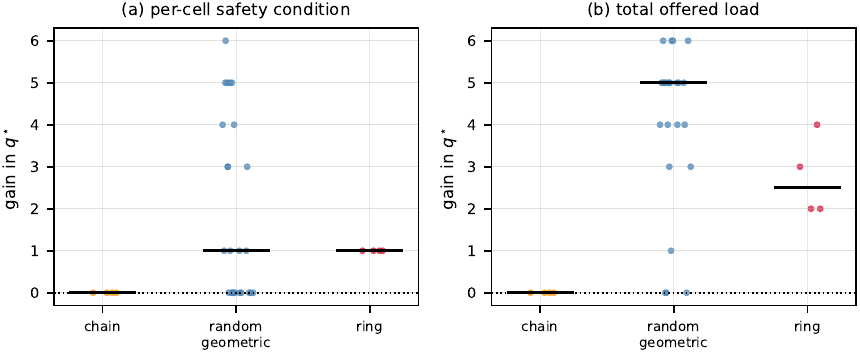}
\caption{The safety-map gain as a distribution rather than a value, over
$\SimEnsembleConfigs$ configurations in three topology families. Each point
is one deployment and the bar is the family median. (a) A per-cell outage
condition. (b) One stated over total offered load. The chain family is the
control: a topology on which the relaxation buys nothing, without which the
sweep would be reporting itself.}
\label{fig:ensemble}
\end{figure}

\begin{table}[!t]
\caption{The distribution of Fig.~\ref{fig:ensemble}, summarised by family.}
\label{tab:ensemble}
\centering
\small
\begin{tabular}{lrccc}
\toprule
family & $n$ & per-cell & aggregate & any gain \\
\midrule
chain & 4 & 0 & 0 & 0\% \\
random geometric & 24 & 1 & 5 & 92\% \\
ring & 4 & 1 & 2 & 100\% \\
\bottomrule
\end{tabular}

\end{table}

\subsection{Influence, in place of a trust flag}
\label{sec:influence}

The division of \S\ref{sec:families} is binary: a channel is forgeable or it
is not. That is the most consequential simplification in the model, and
refining it with a coefficient invented to make a table come out would be
worse than a flag that is honestly crude. A refinement is available without
inventing anything, because the endpoints pin it.

Give channel $m$ an influence coefficient $b_m \ge 0$ bounding how far the
adversary can move it. Then $b_m = 0$ is an uninfluenceable counter,
$b_m = \infty$ is a forgeable channel, and $0 < b_m < \infty$ is the case the
binary split cannot express at all: an aggregate computed honestly from
inputs the adversary chose, which is neither forgeable nor trustworthy.
Channels with $b_m = \infty$ constrain nothing and therefore behave exactly
as deleted rows, while the finite ones contribute their budget to what is
formally extra noise, so \eqref{eq:diam} becomes
\begin{equation}
  \|H\delta\| \;\le\;
    \frac{2\,\|b_{\mathrm{fin}} + v\|_2}{\gamma^H(\Gam_\infty)},
  \label{eq:influence}
\end{equation}
with $\Gam_\infty$ the forgeable set. At a two-point profile
\eqref{eq:influence} is \eqref{eq:diam}, so the binary model is the
restriction of this one rather than an approximation to it; we check that
recovery at $\SimInflControls$ profiles and it holds at every one
(\SimInflRecovers).

Table~\ref{tab:influence} sweeps the graded case. What changes
operationally is the question an operator asks. The budget stops being an
integer count of channels and becomes a bound on a vector: on the reference
plant the architecture tolerates a uniform influence of
$\SimInflBudget$ across every handset-derived channel, and fails at
$\SimInflBad$. That is a specification an operator can write down --- how
far a reported aggregate may be pushed --- where ``how many channels may be
corrupt'' was not.

\begin{table}[!t]
\caption{Graded influence on the handset-derived channels. The first row is
the uninfluenceable case and the last is the forgeable one; the rows between
are what the binary split cannot express.}
\label{tab:influence}
\centering
\small
\begin{tabular}{cccc}
\toprule
UE influence $b$ & $\gamma^H$ & diameter & admissible? \\
\midrule
0 & 6.360 & 0.0009 & yes \\
$1.0\times 10^{-4}$ & 6.360 & 0.0010 & yes \\
$1.0\times 10^{-3}$ & 6.360 & 0.0016 & yes \\
$1.0\times 10^{-2}$ & 6.360 & 0.0085 & yes \\
$1.0\times 10^{-1}$ & 6.360 & 0.0778 & yes \\
$1.0\times 10^{0}$ & 6.360 & 0.7711 & \textbf{no} \\
\bottomrule
\end{tabular}

\end{table}

\subsection{Placement, as an optimisation with a guarantee}
\label{sec:placement}

Corollary~\ref{cor:design} says trusted measurement must render the
safety-relevant subspace observable on its own, and \S\ref{sec:confined}
makes that the confinement threshold. Neither answers the question an
operator asks, which is where the next counter goes.

By \eqref{eq:spectral}, $(A, C_T)$ is observable exactly when no
eigen-direction is invisible to $T$, so a trusted set is feasible precisely
when it intersects $\mathrm{supp}(Cv)$ for every eigen-direction $v$.
Crossing the confinement threshold at least cost is therefore a
\emph{minimum hitting set} problem. That is NP-hard in general, which is
itself the explanation for the surprise this paper reported: intuition about
spreading sensors out fails because the problem is combinatorial coverage
and not geometry.

The coverage function is monotone and submodular, so greedy carries the
standard $1 + \ln n$ guarantee. We check the submodularity rather than cite
it --- $\SimPlaceSubPairs$ sampled pairs, no violation (\SimPlaceSubmodular)
--- and score greedy against the exhaustive optimum on every row reported
rather than against the bound. Greedy attains the optimum on all
$\SimPlaceRows$ configurations, at a worst ratio of
$\SimPlaceWorstRatio$, selecting $\SimPlaceGreedySize$ counters where the
optimum needs $\SimPlaceOptSize$.

The optimiser also explains the paper's own counterintuitive measurement.
Two \emph{adjacent} trusted counters beat two spread ones on this ring
because adjacency is what the hitting-set optimum selects here, not because
adjacency is good. On a different graph the optimum sits elsewhere, which is
why the finding must be computed and never generalised.

\section{Detection scaling}
\label{sec:refuted}

Write $\Delta_{\min}$ for the smallest state deviation along the least
observable direction that the monitor detects with a miss rate of at most
$0.05$ at a false-alarm rate of $0.05$. We predicted that $\Delta_{\min}$
would scale as $\sigma/\gamma_{2q}$, so that the product
$\Delta_{\min}\gamma_{2q}$ would be constant across configurations. Measured
that way it is not: the values are $\{\SimDeltaMinProducts\}$, a coefficient
of variation of $\SimDeltaMinCV$. The apparent failure is a property of the
measurement rather than of the plant, and two properties of the measurement
account for it.

The first is that $\gamma_{2q}$ is not the margin this statistic sees. The
monitor thresholds the decoded infinity norm and the decoder removes $q$
channels, so the spread of that statistic under the null is governed by the
conditioning of the reduced observability map at a $q$-removal, which we
write $\gamma_q$. The quantity $\gamma_{2q}$ governs something else: whether
two attack hypotheses are separable at all, which is the subject of
Theorem~\ref{thm:one}. Pairing a detection threshold against it pairs a
statistic with a margin its estimator never uses. Across our configurations
$\gamma_q$ ranges only over $[\SimGammaQLow, \SimGammaQHigh]$ while
$\gamma_{2q}$ spans a factor of $\SimGammaSpan$, which is the whole of the
apparent saturation.

The second is that the statistic responds to $\Delta\|z\|_\infty$ rather than
to $\Delta$. The least observable direction is a unit vector in the Euclidean
norm, and how much of it any one coordinate carries varies between
configurations, so two systems with identical margins and differently shaped
hidden directions give different $\Delta_{\min}$ for a reason that has
nothing to do with observability.

Correcting both, the product $\Delta_{\min}\gamma_q\|z\|_\infty$ takes the
values $\{\SimCorrectedProducts\}$, a coefficient of variation of
$\SimDeltaMinCVCorrected$ against $\SimDeltaMinCV$ for the pairing with
$\gamma_{2q}$, a factor of $\SimCVImprovement$. The scaling law holds; the
apparent failure was the choice of denominator.

We report the mis-specified pairing beside the corrected one because a
prediction that fails should be checked against the estimator that measured
it before the failure is attributed to the system, and the specific error
here -- pairing a detection threshold with a margin its estimator does not
use -- is one that any similar evaluation can make.

$\Delta_{\min}$ is located by bisection to a relative resolution of
$\SimDeltaMinRes$ per cent, with the median over three seeds whose spread
does not exceed $\SimDeltaMinSeedSpread$ per cent, and the null hypothesis
carries the masking signal so that the test separates deviation size
rather than the presence of an attack. A coarse geometric grid ($28$ per
cent spacing) or a null drawn without the attacker would each manufacture
part of the apparent flatness on its own.

\subsection{A monitor that reads only messages}
\label{sec:trust}

A live proposal for governing learned agents in the radio access
network~\cite{chergui2026agents} scores each agent by how well its asserted
belief matches what its peers expect it to believe. The design advice
attached to that proposal is correct, and it is Corollary~\ref{cor:design}
stated informally: at least one agent per loop should be anchored to
verifiable hardware counters, to close the loop against physical ground
truth. What follows is the matching negative result, and it bounds a class
rather than criticising one construction.

\begin{proposition}[Message-based trust has budget zero]
\label{prop:trust}
Let a monitor's statistic be a function of agent messages and of models of
those agents, and of nothing else. Then every input it reads lies inside the
adversary's reach, so it has no network-derived channel,
Corollary~\ref{cor:design} gives it $q^\star = 0$, and
Corollary~\ref{cor:dichotomy} applies to it unconditionally.
\end{proposition}

We implemented such a metric over $\SimTrustAgents$ agents subscribed to
overlapping channel subsets and ran the construction of \S\ref{sec:attack}
against it. The monitor reads $\SimTrustNetChannels$ network-derived
channels. Every agent reports the correct inference from what it saw and
nobody misreports, so every pairwise divergence stays at the floor: the
lowest trust weight over the whole run is $\SimTrustMinWeight$ while the
state leaves the safe set at step $\SimTrustEscapeStep$ and reaches
$\SimTrustPeak$.

The control matters more than the demonstration. The same metric does catch
an agent that is simply wrong: perturbing one agent's report drops its weight
from $\SimTrustWeightConsistent$ to $\SimTrustWeightWrong$
($\SimTrustControlOk$). The implementation is not inert, and the result is
therefore about the attack rather than about the code. Consistency is not
correctness, and a trust score rewards exactly the property this construction
manufactures.

\subsection{The budget is not robust to identification}
\label{sec:identified}

$q^\star$ is the number this architecture is sold on, and it is computed by
an exact rank test. Exact rank deficiency is \emph{not generic}: the hidden
subspaces on a symmetric ring exist because of the symmetry, and any
perturbation destroys them.

Fitting $(A,C)$ from a $\SimIdSamples$-sample trace of the reference plant's
own output, at a relative state-matrix error of $\SimIdErrA$, moves the
reported budget from $\SimBudgetExactTrue$ to $\SimBudgetExactFitted$, and it
stays moved out to fifteen per cent error. The margin does not move: at the
same fit $\gamma_{2q}$ goes from $\SimMarginTrue$ to $\SimMarginFitted$.

\begin{figure*}[!t]
\centering
\includegraphics[width=\textwidth]{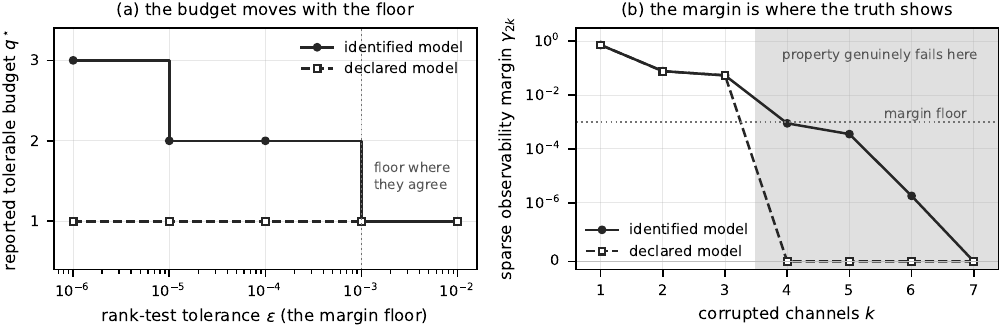}
\caption{Why a budget without a floor is not reproducible. (a) The
tolerable budget the exact rank test reports on a model fitted from the
plant's own output, against the tolerance at which the rank is taken. The
declared model answers $\SimBudgetFloorTrue$ at every tolerance. The
fitted model answers $\SimBudgetExactFitted$ at the tightest one, falls to
two across the middle of the range, and reaches the declared answer only at
$\SimMarginFloor$. The reported budget is thus a property of
the floor as much as of the plant. (b) The same two models read through the
margin instead. Where the property genuinely holds the two agree to
\SimIdMarginAgreement\% and the curves are indistinguishable;
where it genuinely fails the declared margin is exactly zero, plotted on
its own tick below the decades rather than clamped to a small positive
value it does not take. The fitted model reports a positive margin for all
$\SimIdSpuriousK$ of those channel counts, the largest being
$\SimIdWorstSpurious$ --- false assurances that a floor at
$\SimMarginFloor$ rejects and a floor at $10^{-6}$ accepts.}
\label{fig:identified}
\end{figure*}

Fig.~\ref{fig:identified} separates the two readings. The distinction
is not a matter of numerical hygiene. A rank test answers a question about
an exact algebraic property, and a model fitted from data has that property
with probability zero, so the test answers about the fit rather than about
the plant. The margin answers a question about distance, which survives
perturbation. Panel~(b) shows the consequence: the fitted model reports a
positive margin at $\SimIdSpuriousK$ channel counts where the true margin is exactly
zero, and the only thing separating those from the genuine ones is where
the floor is placed.

The operational quantity is therefore the budget against a stated margin
floor, not against a rank test. A direction observable with margin $10^{-9}$
is hidden underneath any real measurement noise, and an adversary does not
care that it is technically observable. Taking the floor at
$\SimMarginFloor$, the fitted model recovers $\SimBudgetFloorFitted$ against
the true $\SimBudgetFloorTrue$. \textbf{Every published budget must name the
floor it was computed at, or it is not reproducible on a fitted model.} This
is a correction to the central computable quantity of \S\ref{sec:thm1} and it
applies to every table in this paper.

\section{The truthful adversary, run}
\label{sec:t1}

Tier T1 is the adversary that makes this problem worth posing and no result
above is about it. It owns devices, induces real conditions and forges
nothing, so the monitor reads the truth and the architecture is said to
contain it. That claim has never been tested, and it is the experiment a
security reviewer asks for first.

We run it on the event-driven queueing loop of~\cite{companion:delay}, where
sessions are integers, arrivals are Poisson and cells block above capacity,
with the adversary's devices offering genuine additional load in one cell.

The claim holds, with a bound. The architecture holds the safe set for
induced loads up to $\SimNLTOneHeld$ additional sessions per period and
loses it at $\SimNLTOneLost$. \emph{The loss is not a monitoring failure.}
The monitor sees the truth throughout and trips correctly; what fails is the
fallback's authority against a load large enough to exceed it. That is a
capacity limit of the verified controller, which is a different and much
better-understood problem than an undetectable one, and it is exactly what
the papers claim when they say the architecture contains T1.

Reporting the bound matters more than reporting the claim. ``Contained''
without a number invites a reader to assume containment at any intensity,
and the architecture does not offer that.

\section{Falsifiability contracts}
\label{sec:falsifiers}

Each result is stated with the observation that would refute it, and each
contract is executed rather than asserted. Reporting that nothing was
observed is weaker than reporting how hard the search looked, so every
contract below corresponds to a search for its own counterexample, and the
budget spent failing to find one is quoted with the status. Each search
additionally runs a \emph{planted} counterexample, so a contract that could
not fail is itself caught.

One contract is not directly searchable and is marked as such. The
Theorem~\ref{thm:one} falsifier quantifies over all monitors, so what is
searched instead is the agreement between the threshold and the
construction: a configuration in which no hidden direction exists below the
threshold, or one exists above it, would refute the construction even if it
left the theorem standing.

\textit{Theorem~\ref{thm:one}.} Exhibit a plant that is not $2q$-sparse
observable together with a monitor reading only those channels that
nonetheless preserves safety against every $q$-sparse attack. Present status:
across $42$ configurations spanning three cell counts and five
trusted-channel counts, the threshold and the construction agree in every
one, and for the configurations tested the construction produces two
trajectories whose measurement sequences agree to within $\SimLeakage$ of
their on-support magnitude on every channel the adversary does not own.

\textit{Theorem~\ref{thm:safety}.} Exhibit a plant that is $H$-sparse
observable at level $2q$ and on which a $q$-sparse attack nonetheless drives
the safety output across its boundary undetected; or, in the other
direction, a safety map at which the condition disagrees with
Definition~\ref{def:sparseobs} when $H = I$. Present status: agreement at
$H = I$ holds at every level in both the boolean and the margin, and
$\SimSafetyWitnesses$ of $\SimSafetyMaps$ maps are witnesses that the
condition is strictly weaker.

\textit{Corollary~\ref{cor:setvalued}.} Exhibit a consistent candidate set
and a true state outside the trigger set on which the set-valued monitor
stays quiet. Present status: $4{,}000$ candidate sets, no fail-open. The
planted case is a monitor thresholding the mean of the set rather than its
supremum, which does fail open on the same draws.

\textit{The narrowed budget claim of \S\ref{sec:identified}.} Exhibit a fit
at which the exact rank test and the margin move together, which would mean
the correction is unnecessary. Present status: at a relative state-matrix
error of $\SimIdErrA$ the rank test moves from $\SimBudgetExactTrue$ to
$\SimBudgetExactFitted$ while the margin moves from $\SimMarginTrue$ to
$\SimMarginFitted$. The planted case is the same test at a floor of zero,
where the two must and do disagree.

\textit{The detection scaling law of \S\ref{sec:refuted}.} Exhibit a
configuration in which $\Delta_{\min}\gamma_q\|z\|_\infty$ departs materially
from the value it takes elsewhere, which would refute the corrected law as
the original pairing refuted the first one. Present status: the product is
$\{\SimCorrectedProducts\}$ across a $\SimGammaSpan$-fold span of
$\gamma_{2q}$, a coefficient of variation of $\SimDeltaMinCVCorrected$.

\textit{Proposition~\ref{prop:trust}.} Exhibit a monitor reading only agent
messages that nonetheless detects the construction of \S\ref{sec:attack}.
Present status: over the whole run the lowest trust weight any agent
receives is $\SimTrustMinWeight$ while the state leaves the safe set. The
planted case is an agent that is simply wrong rather than consistently
deceived, whose weight the same metric does drop, so the result is about the
attack and not about an inert implementation.
\section{Limitations}

The model is a linearisation of a loop that is nonlinear, event-driven and
stochastic in its arrivals. All three results here are structural
properties --- observability, distinguishability and reachability of a
hidden direction --- and structural properties of a linearisation are the
usual currency of control-theoretic security arguments. No quantitative
performance figure survives the objection, which is why we report none.

The decoder used in our study enumerates candidate attack supports and is
exponential in the channel count. That single choice costs us twice. It is
not implementable within a near-real-time budget as written, so our results
concern what is recoverable rather than what is fast, and conflating the two
would be claiming an engineering result we have not produced; convex and
event-triggered alternatives~\cite{shoukry2016event,shoukry2017smt} exist
and would need separate treatment. It also confines every sweep here to
small deployments, which is why the budgets, the threshold at three trusted
counters and the factor of three are all measured on one six-cell ring. The
decision procedure of \S\ref{sec:spectral}, built on the eigenstructure as
the tests of~\cite{montanari2025pbh,montanari2022functional} are, is not
exponential and lifts that limit for deciding the budget; the tables were
computed by enumeration, so every table here should be read as a
small-deployment measurement, and it is the mechanisms rather than the
integers that we claim.

Assumption~\ref{asm:a3} holds trivially wherever the input matrix is a
scaled identity. Table~\ref{tab:a3} exercises the realistic under-actuated
case on one ring, and where the assumption fails the converse has no
construction and the deployment is safe for a reason this analysis does not
model. A general characterisation of which under-actuated deployments
satisfy it is not given.

The corrupted channel set is fixed over time, following the convention in
secure state estimation. That convention is load bearing rather than
cosmetic: an adversary that rotates which channels it corrupts is treated in
the companion paper~\cite{companion:delay}, where it caps the observation
window structurally and defeats configurations that survive a static
adversary.

Two omissions in the adversary model pull in opposite directions. A real
attacker must keep an observable performance indicator healthy or operations
notices, which would slow the attack and which we do not model, so ours is
faster than a realistic one. Against that, the main construction assumes
exact knowledge of $(A,C)$; \S\ref{sec:attack} measures the degradation and
finds model uncertainty a stronger defence than a naive estimate suggests,
and we say so because it weakens our own result.

The trust split between measurement families is the most consequential
modelling choice after the plant itself, and it is probably wrong in both
directions. What reaches a near-real-time controller is a quantity computed
at the base station rather than a raw handset report, which makes our split
too pessimistic; but an aggregate computed honestly from inputs the adversary
chose is neither forgeable nor trustworthy, which the binary split cannot
express at all. Both point at the same structure, an influence
coefficient per channel rather than a trust flag. \S\ref{sec:influence}
supplies it, pinned at its two endpoints; the coefficients between them are
a property of the deployment that this analysis does not measure.

Finally, the sensor-placement finding of Corollary~\ref{cor:design} is one
topology. That two \emph{adjacent} trusted counters outperformed two spread
ones is real for a six-cell ring and must not be generalised into a placement
heuristic, since sparse observability is combinatorial and a different graph
can invert it. It is reported as evidence that intuition is unreliable here,
not as a rule.

\section{Conclusion}

Runtime assurance is the right shape for admitting learned controllers into
network control, and its guarantee is conditional on a property that its
statements assume rather than require. We call that property the assurance
precondition: sparse observability with respect to the safety-relevant
output. It is computable from the plant and its measurement
map, it yields a requirement on where trustworthy measurement must be
placed, and it is strictly weaker than observability of the full state ---
weak enough that stating an outage condition over the cells that carry the
service, rather than over every cell in the model, triples the number of
measurement channels an adversary may own.

The condition is weaker again once the adversary is placed where it
actually sits. Secure state estimation quantifies over every support
because the corrupted set is taken to be arbitrary; in this architecture it
is not, since the division between handset-derived and network-derived
measurement is fixed before deployment and is not the adversary's to
choose. Confining the support to one family turns the budget question into
a threshold question, and above the threshold observability stops bounding
the budget at all: the decoder discards the untrusted family outright and
every channel in it may be corrupt simultaneously. That gives an operator a
lever with a number attached, since crossing the threshold is a matter of
how many uninfluenceable counters the deployment carries and can be checked
from the deployment graph before anything is built. It also gives this
paper a finding about itself, because the configuration we use throughout
sits below the threshold.

Below the threshold the monitor's position is a frontier rather than a
choice between two failures. A monitor that acts on a member of the
consistent set can be steered; one that declines whenever the set disagrees
pays for the adversary's convenience; and one that acts on the whole set
does neither, because the supremum over consistent candidates bounds the true
state from above. That monitor cannot fail open at any margin or noise level.
What it cannot do is take the switch back: an adversary willing to drive the
state along a hidden direction as far as the trigger radius forces the
fallback whenever it chooses.

The same argument settles a question the architecture is about to face at
scale. Proposals to govern concurrent learned agents by scoring how well each
agent's report matches what its peers expect read no channel outside the
adversary's reach, so their tolerable budget is zero however well they
perform against honest faults. Consistency is not correctness, and an attack
that manufactures consistency is exactly what the construction here produces.

Two cautions belong with the result rather than after it. The tolerable
budget is decided by a rank test and rank deficiency is not generic, so on a
model fitted from measurements the reported budget moves threefold at a
fraction of a per cent of model error while the margin barely moves: a budget
quoted without the margin floor it was computed at is not reproducible on
anyone else's model. And a detection-scaling law that appears refuted when a
detection threshold is paired with a margin its estimator does not use holds
once the pairing is corrected.

The architecture does not eliminate the threat. It converts an undetectable
and unattributable safety failure into a bounded and observable availability
cost, and it states the rate of exchange. What that exchange costs in time,
and how many concurrent applications a controller may host before it can no
longer be paid, is the subject of the companion
paper~\cite{companion:delay}.
\appendices

\section{Proof of \texorpdfstring{Theorem~\ref{thm:one}}{Theorem 1}, necessity}
\label{app:thm1}

The sufficiency direction is given in the body. We give the converse in
full, since the construction rather than the existence argument is what
makes the result operational.

Suppose $(A,C)$ is not $2q$-sparse observable. By definition there is an
index set $\Gam'$ with $|\Gam'| = 2q$ such that $(A, C_{\Gam'^c})$ is
unobservable, so its unobservable subspace $\mathcal{N}$ is non-trivial.
Fix $z \in \mathcal{N}\setminus\{0\}$, so that
\begin{equation}
  C_{\Gam'^c} A^k z = 0 \qquad \text{for all } k \ge 0 .
  \label{eq:hidden}
\end{equation}
Two properties of $\mathcal{N}$ are used and both are standard: it is
$A$-invariant, so \eqref{eq:hidden} persists at every step rather than
holding at one instant, and it is a subspace, so $\alpha z \in \mathcal{N}$
for every scalar $\alpha$.

Partition $\Gam' = \Gam_1 \uplus \Gam_2$ with $|\Gam_i| \le q$, and let
the adversary own $\Gam_1$. Let the true state be
$x_k = x^{\mathrm{ben}}_k + \alpha A^k z$, where $x^{\mathrm{ben}}$ is a
benign trajectory remaining in $\Tset$, and let the adversary write
\begin{equation}
  a_k = -\bigl(C A^k \alpha z\bigr)\big|_{\Gam_1} .
\end{equation}
This is admissible: its support lies in $\Gam_1$, which the adversary owns.

Consider the two candidate explanations of the observed sequence. Against
the true state the residual is $a_k$, supported on $\Gam_1$. Against the
benign state it is
\begin{equation}
  C A^k \alpha z + a_k = \bigl(C A^k \alpha z\bigr)\big|_{\Gam_2},
\end{equation}
where the equality uses \eqref{eq:hidden} to discard every channel outside
$\Gam'$ and the definition of $a_k$ to cancel $\Gam_1$. Its support lies
in $\Gam_2$. Both residuals have support at most $q$, so both states are
consistent with an attack of the assumed budget, and no function of the
observed channels distinguishes them.

By Assumption~\ref{asm:a3} the untrusted controller holds the state on
$x^{\mathrm{ben}} + \alpha A^k z$. Since $\mathcal{N}$ is closed under
scaling, choose $\alpha$ with
$\|x^{\mathrm{ben}}_k + \alpha A^k z\|_\infty > x_{\max}$ while
$\|x^{\mathrm{ben}}_k\|_\infty \le \tau$. The true state leaves $\Sset$
while the benign candidate remains in $\Tset$, and the monitor has no
information with which to prefer the former. A monitor that hands over only
once the state has left $\Tset$ therefore keeps the untrusted controller in
control while the true state leaves $\Sset$. \qed

\section*{Code and data availability}

Every numeric value in this paper is generated from a recorded run rather
than typed by hand. The artifact comprises the implementation, the
experiment drivers, the recorded result files, a register of every claim
tested, including those our own data refuted, and the executable falsifier
searches whose budgets are quoted in \S\ref{sec:falsifiers}. A single
command regenerates every figure and table in this paper from the recorded
runs, and the environment is pinned in a container image.

\ArtifactTool

\ArtifactLocation

\bibliographystyle{IEEEtran}
\bibliography{refs}

\begin{IEEEbiographynophoto}{Yasser Al-Eryani}
received the Ph.D. degree in electrical and computer engineering from the
University of Manitoba, Canada. He has held research and engineering roles
in wireless systems, AI-enabled network architectures, O-RAN, and
distributed intelligent systems. His research interests include information
theory, wireless communications, continuous-aperture systems,
operator-theoretic signal models, AI-native radio access networks, and
geometry-aware communication architectures. He has authored peer-reviewed
publications and patents in wireless communication, AI systems, and network
intelligence.
\end{IEEEbiographynophoto}

\end{document}